\documentclass[12pt,oneside]{amsart}
\usepackage{amssymb}
\usepackage{amsfonts}
\usepackage{amscd}
\usepackage[matrix, arrow, curve]{xy}
\usepackage{graphicx}
\usepackage{color}
\usepackage{young}
\usepackage{youngtab}
\usepackage{titletoc}
\usepackage{extarrows}
\usepackage{varioref}
\usepackage{bm}
\usepackage{ytableau}
\usepackage{tikz-cd}
\usepackage[normalem]{ulem}
\usepackage{hyperref,xcolor}% http://ctan.org/pkg/{hyperref,xcolor}
\definecolor{wine-stain}{rgb}{0.5,0,0}
\hypersetup{
	colorlinks,
	linkcolor=wine-stain,
	linktoc=all
}

\numberwithin{equation}{section}
\usepackage{todonotes}
\definecolor{pistachio}{rgb}{0.58, 0.77, 0.45}
\definecolor{eggshell}{rgb}{0.94, 0.92, 0.84}
\numberwithin{equation}{section}
\newtheorem{theorem}{Theorem}[section]

\newtheorem{proposition}[theorem]{Proposition}
\newtheorem{corollary}[theorem]{Corollary}

\theoremstyle{definition}
\newtheorem{definition}[theorem]{Definition}
\theoremstyle{remark}
\newtheorem*{remark}{Remark}
\newtheorem{example}{Example}

\begin{document}
%%%%%%%%%%%%%%%%%%%%%%%%%%%%%%%%%%%%%%%%%%%%%%%%%
%%%%%%%%%%%%  macrodefinitions
%%%%%%%%%%%%%%%%%%%%%%%%%%%%%%%%%%%%%%%%%%%%%%%%%
%  Macros (general)
%%%%%%%%%%%%%%%%%%%%%%%
%\newcommand{\MgNekp}{\mathcal{M}_{g,N+1}^{(k,p)}} %% moduli space
%\newcommand{\M}{\mathcal{M}_{g,N+1}^{(1)}}
\newcommand{\M}{\mathcal{M}}
\newcommand{\F}{\mathcal{F}}

\newcommand{\Teich}{\mathcal{T}_{g,N+1}^{(1)}}
\newcommand{\T}{\mathrm{T}}
%%%%   temporary
\newcommand{\corr}{\bf}
\newcommand{\vac}{|0\rangle}
\newcommand{\Ga}{\Gamma}
\newcommand{\new}{\bf}
\newcommand{\define}{\def}
\newcommand{\redefine}{\def}
\newcommand{\Cal}[1]{\mathcal{#1}}
\renewcommand{\frak}[1]{\mathfrak{{#1}}}
\newcommand{\Hom}{\rm{Hom}\,}
%%%%%%%%%%%%%%%%%%%%%%%%%%%%%%%%%%%%
%   Referencing Scheme of Martin
%%%%%%%%%%%%%%%%%%%%%%%%%%
\newcommand{\refE}[1]{(\ref{E:#1})}
\newcommand{\refCh}[1]{Chapter~\ref{Ch:#1}}
\newcommand{\refS}[1]{Section~\ref{S:#1}}
\newcommand{\refSS}[1]{Section~\ref{SS:#1}}
\newcommand{\refT}[1]{Theorem~\ref{T:#1}}
\newcommand{\refO}[1]{Observation~\ref{O:#1}}
\newcommand{\refP}[1]{Proposition~\ref{P:#1}}
\newcommand{\refD}[1]{Definition~\ref{D:#1}}
\newcommand{\refC}[1]{Corollary~\ref{C:#1}}
\newcommand{\refL}[1]{Lemma~\ref{L:#1}}
\newcommand{\refEx}[1]{Example~\ref{Ex:#1}}
%%%%%%%%%%%%%%%%%%%%%%%%%%%%%%%%%%
\newcommand{\R}{\ensuremath{\mathbb{R}}}
\newcommand{\C}{\ensuremath{\mathbb{C}}}
\newcommand{\N}{\ensuremath{\mathbb{N}}}
\newcommand{\Q}{\ensuremath{\mathbb{Q}}}
\renewcommand{\P}{\ensuremath{\mathcal{P}}}
\newcommand{\Z}{\ensuremath{\mathbb{Z}}}
%%%%%%%%%%%%%%%%%%%%%%%%%%%%%%%%%%%%%%%%%%
\newcommand{\kv}{{k^{\vee}}}
%%%%%%%%%%%%%%%%%%%%%%%%%%%%%%%%%%%%%%%%%%%%%
\renewcommand{\l}{\lambda}
%%%%%%%%%%%%%%%%%%%%%%%%%%%%%%%%%%%%%%%%%%%%%%%%%%
\newcommand{\gb}{\overline{\mathfrak{g}}}
\newcommand{\dt}{\tilde d}     % Oleg
\newcommand{\hb}{\overline{\mathfrak{h}}}
\newcommand{\g}{\mathfrak{g}}
\newcommand{\h}{\mathfrak{h}}
\newcommand{\gh}{\widehat{\mathfrak{g}}}
\newcommand{\ghN}{\widehat{\mathfrak{g}_{(N)}}}
\newcommand{\gbN}{\overline{\mathfrak{g}_{(N)}}}
\newcommand{\tr}{\mathrm{tr}}
\newcommand{\gln}{\mathfrak{gl}(n)}
\newcommand{\son}{\mathfrak{so}(n)}
\newcommand{\spnn}{\mathfrak{sp}(2n)}
\newcommand{\sln}{\mathfrak{sl}}
\newcommand{\sn}{\mathfrak{s}}
\newcommand{\so}{\mathfrak{so}}
\newcommand{\spn}{\mathfrak{sp}}
\newcommand{\tsp}{\mathfrak{tsp}(2n)}
\newcommand{\gl}{\mathfrak{gl}}
\newcommand{\slnb}{{\overline{\mathfrak{sl}}}}
\newcommand{\snb}{{\overline{\mathfrak{s}}}}
\newcommand{\sob}{{\overline{\mathfrak{so}}}}
\newcommand{\spnb}{{\overline{\mathfrak{sp}}}}
\newcommand{\glb}{{\overline{\mathfrak{gl}}}}
\newcommand{\Hwft}{\mathcal{H}_{F,\tau}}
\newcommand{\Hwftm}{\mathcal{H}_{F,\tau}^{(m)}}

%%%%%%%%%%%%%%%%%%%%%%%%%%%%%%%%%%%%%%%%%%%%%%%%%%%%
\newcommand{\car}{{\mathfrak{h}}}    % Cartan subalgebra
\newcommand{\bor}{{\mathfrak{b}}}    % Borel subalgebra
\newcommand{\nil}{{\mathfrak{n}}}    % nilpotent subalgebra
\newcommand{\vp}{{\varphi}}
\newcommand{\bh}{\widehat{\mathfrak{b}}}  % Borel subalgebra of KN algebra
\newcommand{\bb}{\overline{\mathfrak{b}}}  % Borel subalgebra of KN algebra
\newcommand{\Vh}{\widehat{\mathcal V}}
\newcommand{\KZ}{Kniz\-hnik-Zamo\-lod\-chi\-kov}
\newcommand{\TUY}{Tsuchia, Ueno  and Yamada}
\newcommand{\KN} {Kri\-che\-ver-Novi\-kov}
\newcommand{\pN}{\ensuremath{(P_1,P_2,\ldots,P_N)}}
\newcommand{\xN}{\ensuremath{(\xi_1,\xi_2,\ldots,\xi_N)}}
\newcommand{\lN}{\ensuremath{(\lambda_1,\lambda_2,\ldots,\lambda_N)}}
\newcommand{\iN}{\ensuremath{1,\ldots, N}}
\newcommand{\iNf}{\ensuremath{1,\ldots, N,\infty}}

\newcommand{\tb}{\tilde \beta}
\newcommand{\tk}{\tilde \varkappa}
\newcommand{\ka}{\kappa}
\renewcommand{\k}{\varkappa}
\newcommand{\ce}{{c}}

\newcommand{\Pif} {P_{\infty}}
\newcommand{\Pinf} {P_{\infty}}
\newcommand{\PN}{\ensuremath{\{P_1,P_2,\ldots,P_N\}}}
\newcommand{\PNi}{\ensuremath{\{P_1,P_2,\ldots,P_N,P_\infty\}}}
\newcommand{\Fln}[1][n]{F_{#1}^\lambda}
\newcommand{\tang}{\mathrm{T}}
\newcommand{\Kl}[1][\lambda]{\can^{#1}}
\newcommand{\A}{\mathcal{A}}
\newcommand{\U}{\mathcal{U}}
\newcommand{\V}{\mathcal{V}}
\newcommand{\W}{\mathcal{W}}
\renewcommand{\O}{\mathcal{O}}
\newcommand{\Ae}{\widehat{\mathcal{A}}}
\newcommand{\Ah}{\widehat{\mathcal{A}}}
\newcommand{\La}{\mathcal{L}}
\newcommand{\Le}{\widehat{\mathcal{L}}}
\newcommand{\Lh}{\widehat{\mathcal{L}}}
\newcommand{\eh}{\widehat{e}}
\newcommand{\Da}{\mathcal{D}}
\newcommand{\kndual}[2]{\langle #1,#2\rangle}
\newcommand{\cins}{\frac 1{2\pi\mathrm{i}}\int_{C_S}}
\newcommand{\cinsl}{\frac 1{24\pi\mathrm{i}}\int_{C_S}}
\newcommand{\cinc}[1]{\frac 1{2\pi\mathrm{i}}\int_{#1}}
\newcommand{\cintl}[1]{\frac 1{24\pi\mathrm{i}}\int_{#1 }}
\newcommand{\w}{\omega}
\newcommand{\ord}{\operatorname{ord}}
\newcommand{\res}{\operatorname{res}}
\newcommand{\nord}[1]{:\mkern-5mu{#1}\mkern-5mu:}
\newcommand{\codim}{\operatorname{codim}}
\newcommand{\ad}{\operatorname{ad}}
\newcommand{\Ad}{\operatorname{Ad}}
\newcommand{\supp}{\operatorname{supp}}

%%%%%%%%%%%%%%%%%%%%%%%%%%%%%%%%%%%%%%%%%%%%%%%%
\newcommand{\Fn}[1][\lambda]{\mathcal{F}^{#1}}
\newcommand{\Fl}[1][\lambda]{\mathcal{F}^{#1}}
\renewcommand{\Re}{\mathrm{Re}}

\newcommand{\ha}{H^\alpha}

\define\ldot{\hskip 1pt.\hskip 1pt}
\define\ifft{\qquad\text{if and only if}\qquad}
\define\a{\alpha}
\redefine\d{\delta}
\define\w{\omega}
\define\ep{\epsilon}
\redefine\b{\beta} \redefine\t{\tau} \redefine\i{{\,\mathrm{i}}\,}
\define\ga{\gamma}
\define\cint #1{\frac 1{2\pi\i}\int_{C_{#1}}}
\define\cintta{\frac 1{2\pi\i}\int_{C_{\tau}}}
\define\cintt{\frac 1{2\pi\i}\oint_{C}}
\define\cinttp{\frac 1{2\pi\i}\int_{C_{\tau'}}}
\define\cinto{\frac 1{2\pi\i}\int_{C_{0}}}
%\define\cinttt{\frac 1{24\pi\i}\int_{C_{\tau}}}
\define\cinttt{\frac 1{24\pi\i}\int_C}
\define\cintd{\frac 1{(2\pi \i)^2}\iint\limits_{C_{\tau}\,C_{\tau'}}}
\define\dintd{\frac 1{(2\pi \i)^2}\iint\limits_{C\,C'}}
\define\cintdr{\frac 1{(2\pi \i)^3}\int_{C_{\tau}}\int_{C_{\tau'}}
\int_{C_{\tau''}}}
\define\im{\operatorname{Im}}
\define\re{\operatorname{Re}}
%\define\res{\text{res}}
\define\res{\operatorname{res}}
\redefine\deg{\operatornamewithlimits{deg}}
\define\ord{\operatorname{ord}}
\define\rank{\operatorname{rank}}
\define\fpz{\frac {d }{dz}}
\define\dzl{\,{dz}^\l}
\define\pfz#1{\frac {d#1}{dz}}

\define\K{\Cal K}
\define\U{\Cal U}
\redefine\O{\Cal O}
\define\He{\text{\rm H}^1}
\redefine\H{{\mathrm{H}}}
\define\Ho{\text{\rm H}^0}
\define\A{\Cal A}
\define\Do{\Cal D^{1}}
\define\Dh{\widehat{\mathcal{D}}^{1}}
\redefine\L{\Cal L}
\newcommand{\ND}{\ensuremath{\mathcal{N}^D}}
\redefine\D{\Cal D^{1}}
\define\KN {Kri\-che\-ver-Novi\-kov}
\define\Pif {{P_{\infty}}}
\define\Uif {{U_{\infty}}}
\define\Uifs {{U_{\infty}^*}}
\define\KM {Kac-Moody}
\define\Fln{\Cal F^\lambda_n}
%%%%%%%%%%%%%%%%%%%%
\define\gb{\overline{\mathfrak{ g}}}
\define\G{\overline{\mathfrak{ g}}}
\define\Gb{\overline{\mathfrak{ g}}}
\redefine\g{\mathfrak{ g}}
\define\Gh{\widehat{\mathfrak{ g}}}
\define\gh{\widehat{\mathfrak{ g}}}
%%%%%%%%%%%%%%%%%%%%%%%%%%
\define\Ah{\widehat{\Cal A}}
\define\Lh{\widehat{\Cal L}}
\define\Ugh{\Cal U(\Gh)}
\define\Xh{\hat X}
\define\Tld{...}
\define\iN{i=1,\ldots,N}
\define\iNi{i=1,\ldots,N,\infty}
\define\pN{p=1,\ldots,N}
\define\pNi{p=1,\ldots,N,\infty}
\define\de{\delta}

\define\kndual#1#2{\langle #1,#2\rangle}
\define \nord #1{:\mkern-5mu{#1}\mkern-5mu:}
%\define \MgN{{\Cal M}_{g,N}} %% moduli space
%\define \MgNp{{\Cal M}_{g,N}^{(p)}} %% moduli space
\newcommand{\MgN}{\mathcal{M}_{g,N}} %% moduli space
\newcommand{\MgNeki}{\mathcal{M}_{g,N+1}^{(k,\infty)}} %% moduli space
\newcommand{\MgNeei}{\mathcal{M}_{g,N+1}^{(1,\infty)}} %% moduli space
\newcommand{\MgNekp}{\mathcal{M}_{g,N+1}^{(k,p)}} %% moduli space
\newcommand{\MgNkp}{\mathcal{M}_{g,N}^{(k,p)}} %% moduli space
\newcommand{\MgNk}{\mathcal{M}_{g,N}^{(k)}} %% moduli space
\newcommand{\MgNekpp}{\mathcal{M}_{g,N+1}^{(k,p')}} %% moduli space
\newcommand{\MgNekkpp}{\mathcal{M}_{g,N+1}^{(k',p')}} %% moduli space
\newcommand{\MgNezp}{\mathcal{M}_{g,N+1}^{(0,p)}} %% moduli space
\newcommand{\MgNeep}{\mathcal{M}_{g,N+1}^{(1,p)}} %% moduli space
\newcommand{\MgNeee}{\mathcal{M}_{g,N+1}^{(1,1)}} %% moduli space
\newcommand{\MgNeez}{\mathcal{M}_{g,N+1}^{(1,0)}} %% moduli space
\newcommand{\MgNezz}{\mathcal{M}_{g,N+1}^{(0,0)}} %% moduli space
\newcommand{\MgNi}{\mathcal{M}_{g,N}^{\infty}} %% moduli space
\newcommand{\MgNe}{\mathcal{M}_{g,N+1}} %% moduli space
\newcommand{\MgNep}{\mathcal{M}_{g,N+1}^{(1)}} %% moduli space
\newcommand{\MgNp}{\mathcal{M}_{g,N}^{(1)}} %% moduli space
\newcommand{\Mgep}{\mathcal{M}_{g,1}^{(p)}} %% moduli space
\newcommand{\MegN}{\mathcal{M}_{g,N+1}^{(1)}} %% moduli space

%\define \mpt{(M,P_1,P_2,\ldots, P_N,\Pif)} %% moduli point
%\define \mpp{(M,P_1,P_2,\ldots, P_N)} %% moduli point
%\define \MgNn{{\Cal M}_{g,N}^{(1)}} %% moduli space
%\define \MgNen{{\Cal M}_{g,N+1}^{(1)}} %% moduli space
%\define \Mgo{{\Cal M}_{g,0}} %% moduli space
%\define \mptn{(M,P_1,P_2,\ldots, P_N,\Pif,z_1,\ldots,z_N,z_\infty)}
 %% moduli point
%\define \mppn{(M,P_1,P_2,\ldots, P_N,z_1,\ldots,z_N)} %% moduli point
\define \sinf{{\widehat{\sigma}}_\infty}
\define\Wt{\widetilde{W}}
\define\St{\widetilde{S}}
\newcommand{\SigmaT}{\widetilde{\Sigma}}
\newcommand{\hT}{\widetilde{\frak h}}
\define\Wn{W^{(1)}}
\define\Wtn{\widetilde{W}^{(1)}}
\define\btn{\tilde b^{(1)}}
\define\bt{\tilde b}
\define\bn{b^{(1)}}
\define \ainf{{\frak a}_\infty} %matrices with a finite number of
                                %diagonals

%
%%%%%%%%%% Olegs definitions %%%%%%%%%%%%%%%%%%%%%%%%%%%%%%%%%%%
\define\eps{\varepsilon}    % Oleg
\newcommand{\e}{\varepsilon}
\define\doint{({\frac 1{2\pi\i}})^2\oint\limits _{C_0}
       \oint\limits _{C_0}}                            % Oleg
\define\noint{ {\frac 1{2\pi\i}} \oint}   % Oleg
\define \fh{{\frak h}}     % Oleg
\define \fg{{\frak g}}     % Oleg
\define \GKN{{\Cal G}}   % affine Krichever-Novikov algebra % Oleg
\define \gaff{{\hat\frak g}}   % affine Krichever-Novikov algebra
\define\V{\Cal V}
\define \ms{{\Cal M}_{g,N}} %% moduli space
\define \mse{{\Cal M}_{g,N+1}} %% moduli space
%%%%%%%%%%%%%%%%%%%%%%%%%%%%%%%%%%%%%%
\define \tOmega{\Tilde\Omega}
\define \tw{\Tilde\omega}
\define \hw{\hat\omega}
\define \s{\sigma}
\define \car{{\frak h}}    % Cartan subalgebra
\define \bor{{\frak b}}    % Borel subalgebra
\define \nil{{\frak n}}    % nilpotent subalgebra
\define \vp{{\varphi}}
\define\bh{\widehat{\frak b}}  % Borel subalgebra of KN algebra
\define\bb{\overline{\frak b}}  % Borel subalgebra of KN algebra
\define\KZ{Knizhnik-Zamolodchikov}
\define\ai{{\alpha(i)}}
\define\ak{{\alpha(k)}}
\define\aj{{\alpha(j)}}
\newcommand{\calF}{{\mathcal F}}
\newcommand{\ferm}{{\mathcal F}^{\infty /2}}
\newcommand{\Aut}{\operatorname{Aut}}
\newcommand{\End}{\operatorname{End}}
%%%%%%%%%%%%%%%%%%%%%%%%%%%%%%%%%%%%%%%%%%%
%%%%%%%%%%%%%%%%  öâåò %%%%%%%%%%%%%%%%%%%%%%%%%%%%%%%%%
\newcommand{\novoe}{{\color[rgb]{1,0,0}\bf (Íîâîå)}}
\newcommand{\staroe}{{\color[rgb]{0,0,1}\bf (Ñòàðîå)}}
\newcommand{\red}{\color[rgb]{1,0,0}}
\newcommand{\blue}{\color[rgb]{0,0,1}}
\newcommand{\viol}{\color[rgb]{1,0,1}}%%%%%%%%%%%%%%%%%%%%%%%%%%%%%%%%%%%%%%%%%%%%%%

%%%%%%%%%%%%%%%%%%%%%%%%%%%%%%%%%
%%%%%%%%%%%%%%   for laxcent
%%%%%%%%%%%%%%%%%%%%%%%%%%%%%%%%%%
\newcommand{\laxgl}{\overline{\mathfrak{gl}}}
\newcommand{\laxsl}{\overline{\mathfrak{sl}}}
\newcommand{\laxso}{\overline{\mathfrak{so}}}
\newcommand{\laxsp}{\overline{\mathfrak{sp}}}
\newcommand{\laxs}{\overline{\mathfrak{s}}}
\newcommand{\laxg}{\overline{\frak g}}
\newcommand{\bgl}{\laxgl(n)}
%%%%%%%%%%%%%%%%%%%%%%%%
\newcommand{\tX}{\widetilde{X}}
\newcommand{\tY}{\widetilde{Y}}
\newcommand{\tZ}{\widetilde{Z}}
%%%%%%%%%%%%%%%%%%%%%%%%%%%%%%%%%%%%%%%%%%
%%%%%%%%%%%%  END of macrodefinitions
%%%%%%%%%%%%%%%%%%%%%%%%%%%%%%%%%%%%%%%%%
%%%%%%%%%%%%%%%%%%%%%%%%%%%%%%%%%%%%%%%%%%%%%%%%%%% Bin's new operatornames
\newcommand{\GL}{\operatorname{GL}}
\newcommand{\PGL}{\operatorname{PGL}}
\newcommand{\SL}{\operatorname{SL}}
\newcommand{\Sp}{\operatorname{Sp}}
\newcommand{\SO}{\operatorname{SO}}
\newcommand{\SU}{\operatorname{SU}}
\newcommand{\pdeg}{\text{par-}deg}
\newcommand{\gr}{\operatorname{gr}}
%%%%%%%%%%%%%%%%%%%%%%%%%%%%%%%%%

%%%%%%%%%%%%%%%%%%%%%%%%%%%%%%%%%
%Top-Matter
%%%%%%%%%%%%%%%%%%%%%%%%%%%%%%%
%%%%%%%%%%%%%%%%%    private header  %%%%%%%%%%%%%%%%%%%%

%\large{
\title[]{Hitchin systems: some recent advances}
\author[O.K.Sheinman]{O.K.Sheinman}
%\date{\today}
\thanks{This work was performed at the Steklov International Mathematical Center and supported by the Ministry of Science and Higher Education of the Russian Federation (agreement no. 075-15-2022-265)}
\address{Steklov Mathematical Institute of the Russian Academy of Sciences}
\author[Bin Wang]{Bin Wang}
%\date{\today}
%\thanks{}
\address{Steklov Mathematical Institute of the Russian Academy of Sciences}
\dedicatory{}
\maketitle
\begin{abstract}
A survey of some recent advances in parabolic Hitchin systems (parabolic Bouville--Narasimhan--Ramanan correspondence, mirror symmetry for parabolic Hitchin systems), and in exact methods of solving the non-parabolic Hitchin systems.
\end{abstract}
\tableofcontents
%%%%%%%%%%%%%%%%%%%%%%%%%%%%%%%%%%%%%%%%
\section{Introduction}
The integrable systems, this survey is devoted to, are invented by N.Hitchin \cite{Hitchin} and bear his name. They are attractive in their geometric simplicity, and have been studied mainly from geometrical point of view. However methods of solving them are much less developed. In this survey we will touch on both issues, however we don't set ourselves an impossible task of a comprehensive review (several mathematical schools and centers are working on the subject). We focus on two fields of our own research, namely on parabolic Hitchin systems and on exact methods of solving the Hitchin systems. In the introductory \refS{base} we give the necessary definitions for classical Hitchin systems, and briefly describe their main applications.

Parabolic Hitchin systems are invented in the series of works of C.S. Seshadri with coauthors. While original Hitchin systems are related to  compact Riemann surfaces, the parabolic ones correspond to Riemann surfaces with punctures, and some additional structure referred to as a parabolic structure. Parabolic Hitchin systems are the subject of \refS{Parab} of the present paper.

We begin with a general definition of a parabolic structure as a set of filtrations and weights associated with punctures. As a motivation, we present the original approach by Mechta and Seshadri where the parabolic structure is related with a unitary representation of a parabolic subgroup (of the fundamental group of the Riemann surface) at a cusp. Then we define parabolic vector bundles as vector bundles endowed with the parabolic structure, and parabolic Higgs bundles as parabolic bundles endowed with parabolic Higgs fields. The latter are required to satisfy certain nilpotency conditions with respect to the filtrations at punctures. Then the theory of parabolic Hitchin systems is being developed by the same scheme as of original Hitchin systems, with modifications caused by the presence of the parabolic structure, especially the nilpotency conditions.

It is a fundamental problem of the theory of integrable systems to investigate the fibers of the corresponding Lagrangian foliations. For the classic integrable systems these are Liouville tori. In the context of Hitchin systems the fibers are level varieties of the Hitchin map given by evaluation of a full set of independent Hamiltonians on Higgs bundles. Typical fibers of a Hitchin map are proved to be Jacobians of the corresponding spectral curves in the case of structure group $GL(n)$, and to be the Prym varieties (Prymians) of spectral curves in the case of simple classical structure groups.

We discuss here the following two results on the fibers of Hitchin maps: the parabolic BNR correspondence, and the mirror symmetry for parabolic Hitchin systems.

The BNR (Bouville-Narasimhan-Ramanan) correspondence \cite{BNR89} is the correspondence between Higgs bundles in the fibers of Hitchin maps and line bundles on the corresponding spectral curves. Following \cite{BNR89} we give the BNR correspondence for parabolic Hitchin systems which requires a certain accuracy in the approach to spectral curves. In the context of Inverse Spectral Method (\refSS{ISM}) the above  line bundles on spectral curves emerge as the eigenspace bundles of Lax operators. The Baker--Akhieser vector-function defined in \refSS{Spec_tr} is nothing but a section of such line bundle.

Coming to mirror symmetry for Hitchin systems, first notice that geometrically  it  is expressed by Donagi and Pantev theorem \cite{DP12}  stating that the fibers of two Hitchin systems with Langlands dual structure groups are torsors over dual Abelian varieties (see \refT{DonPan} below). It can be proved to be an (Abelian) implication of the homological mirror symmetry in sense of Kontsevich. Here, we keep the concept of topological mirror symmetry interpreted as equality of stringy Hodge numbers, or equivalently stringy E-polynomials (\refSS{TopMirSym}) corresponding to the parabolic Hitchin systems with the same base curve and Langlands dual structure groups (Theorems \ref{conjecture}, \ref{main theorem}).

It is another goal of the present survey to discuss the exact methods of solving the Hitchin systems. It is the subject of \refS{Exact}. The systems with the structure group $GL(n)$ can be explicitly resolved in terms of theta fuctions both by means the inverse spectral method  (Krichever, \cite{Kr_Lax}), and by means the method of separation of variables (A. Gorsky, N. Nekrasov, V. Rubtsov \cite{GNR}, O.Sheinman \cite{Sh_FAN_2019,Sh_ProcStekl_2020}). In frame of separation of variables, we give an explicit theta function formula for solutions which was not published elsewhere (\refSS{theta_sol}). It is pretty clear that the separation of variables technique is applicable also in the case of parabolic Hitchin systems with the structure group $GL(n)$. We hope to revisit this subject in the future.  As we can see, to show the mirror symmetry, it is crucial to interpret parabolic Higgs bundles in the fibers of Hitchin maps as line bundles on the normalization of the corresponding spectral curves. %\footnote{\edit{Considering the analytic nature of the separation of variables, it seems that the p-adic integration cannot be applied then. But the endoscopic decomposition might still work. See discussions below Theorem 3.19.}}
It would be also interesting to obtain such an interpretation and comprehend the mirror symmetry for Hitchin systems in terms of separation of variables.

In the case of systems with simple structure groups, no matter parabolic or non-parabolic, obstructions emerge for both methods. It turns out to be that the Baker--Akhieser function of the system has a dynamical pole divisor in this case, which makes the inverse spectral method inapplicable, at least in the currently known form. As for the separation of variables, the method can be developed quite far, up to finding out the action-angle coordinates in many cases. However, any analog of the theta function formula of \refSS{theta_sol} can not be obtained in general due to the peculiarity of the inversion problem on Prymians. We hope that identifying this problem will serve as an incentive for further research in this direction.

%%%%%%%%%%%%%%%%%%%%%%%%%%%%%%%%%%%%%%%%
\section{Basic results on Hitchin systems}\label{S:base}
\subsection{Hitchin systems: the definition (N.Hitchin \cite{Hitchin})}
Assume, $\Sigma$ is a compact genus $g$ Riemann surface with a conformal structure, $G$ is a complex semisimple or reductive Lie group, $\g = Lie(G)$, $P_0$ is a smooth principal $G$-bundle on~$\Sigma$.

By \emph{holomorphic structure} on $P_0$ we mean a $(0, 1)$-connection, i.e. a differential operator on the sheaf of sections of the bundle $P_0$ locally given as $\bar{\partial} + \w$, where $\w \in \Omega^{0, 1}(\Sigma, \g)$, and, under the action of a gluing function $g$, $\w$ transforms as follows:
\[
\w \to g\w g^{-1} - (\bar{\partial}g)g^{-1}.
\]

Suppose, $\A$ is a space of semistable \cite{Hitchin} holomorphic structures on $P_0$, ${\mathcal G}$ is a group of smooth global gauge transformations. The quotient ${\mathcal N} = \A/{\mathcal G}$ is referred to as the \emph{moduli space of semi-stable holomorphic structures on $P_0$}. We will define a Hitchin system as a dynamical system with ${\mathcal N}$ as a configuration space, and $T^{*}({\mathcal N})$ as a phase space. A point in ${\mathcal N}$ is a principal holomorphic $G$-bundle on $\Sigma$. In the case that $G$ is semisimple we have $\dim{\mathcal N} = \dim\g\cdot (g-1)$. We will consider $G=GL(n)$ as a typical representative of the class of reductive groups. For $G=GL(n)$ we have $\dim{\mathcal N} = n^2\cdot (g-1)+1=\dim\g\cdot (g-1)+1$.

As defined above, the \emph{phase space} of a Hitchin system is $T^{*}({\mathcal N})$. According to Kodaira--{\blue}Spencer theory, $T_P({\mathcal N}) \simeq H^1(\Sigma, \Ad P)$. By Serre duality
\[
  T^*_P({\mathcal N})\simeq H^0(\Sigma, \Ad P\otimes {\mathcal K})
\]
where ${\mathcal K}$ is the canonical class of $\Sigma$, $\Ad P\otimes {\mathcal K}$ is a holomorphic vector bundle with a fiber $\g\otimes\C$. We denote the points of $T^*({\mathcal N})$ by $(P, \Phi)$, where $P \in {\mathcal N}, ~\Phi\in H^0(\Sigma, \Ad P\otimes {\mathcal K})$. Sections of the sheaf $T^*({\mathcal N})$ are called {\it Higgs fields}.

Assume, $\chi_i$ is a homogeneous invariant polynomial on $\g$ of degree $d_i$. It defines a map $\chi_i(P)\colon H^0(\Sigma, \Ad P\otimes {\mathcal K})\to H^0(\Sigma, {\mathcal K}^{\otimes d_i})$ for each $P \in {\mathcal N}$. Let $\Phi$ stay for a Higgs field, then we can define $\chi_i(P, \Phi) = (\chi_i(P))(\Phi(P))$. By that, to each point $(P, \Phi)$ of the phase space we have assigned an element of $H^0(\Sigma, {\mathcal K}^{\otimes d_i})$. Suppose $\{\Omega^i_j\}$ is a basis in $H^0(\Sigma, {\mathcal K}^{\otimes d_i})$, then $\chi_i(P, \Phi) = \sum_j H_{i, j}(P,\Phi)\Omega^i_j$, where $H_{i, j}(P, \Phi)$ are scalar functions on $T^*({\mathcal N})$. For any $i$, $j$ the function $H_{i, j}(P, \Phi)$ is called a {\it Hitchin's Hamiltonian}.
\begin{theorem}[\cite{Hitchin}]\label{T:Hitch}
Whatever be a choice of the base $\{\Omega^i_j\}$, Hitchin Hamiltonians form a complete set of Hamiltonians in involution on $T^*({\mathcal N})$ with respect to the natural symplectic structure.
\end{theorem}
The set $d_1,\ldots,d_n$ is a characteristic invariant of the Lie algebra $\g$, hence the space $\mathcal{H}=\bigoplus_{i=1}^nH^0(\Sigma, {\mathcal K}^{\otimes d_i})$ ($n=\rank\g$) is unambiguously defined. It is referred to as the Hitchin base space. Let $\chi=\{\chi_1,\ldots,\chi_n\}$ be a base in the space of $\Ad G$-invariant polynomials on $\g$. Then $\chi$ defines the map $h_\chi :\, T^*(\mathcal{N})\to \mathcal{H}$ called Hitchin map. Though the last depends on the choice of the base $\chi=\{\chi_1,\ldots,\chi_n\}$, its level varieties are again defined unambiguously. They are referred to as Hitchin fibers. It is one of our main goals in this survey, to describe generalizations of the Hitchin map, and give the description of the corresponding Hitchin fibers. It is also a goal to provide effective alternatives to the aforementioned definition of Hitchin systems, along with methods for obtaining exact solutions for them.

%%%%%%%%%%%%%%%%%%%%%%%%%%%%%%%%%%%%%%%%%%%%%%%%%%%%%%%%%%%%%
\subsection{Classical applications}\label{SS:class_app}

The first, and most known application of Hitchin systems is given by N.Hitchin himself in \cite{Flat_conn}. It is the application to the 2-dimensional Conformal Field Theory (below 2dCFT). By one of the equivalent definitions \cite{FrSh}, a 2dCFT is a projectively flat connection in a holomorphic vector bundle on the moduli space of curves (with punctures in general). Then a problem of constructing such kind of connection emerges. Recall that a connection $\nabla$ is said to be projectively flat if there is a 2-cocycle $\l$ on the Lie algebra of smooth vector fields such that $[\nabla_X,\nabla_Y]=\l(X,Y)\cdot \mathbf{1}$ for any pair of  smooth vector fields $X,Y$, where $\mathbf{1}$ is the identical operator.

%Let $\Sigma$ be a real compact surface, $M=Hom(\pi_1(\Sigma),SU(n))/SU(n)$ be the variety of equivalence classes of irreducible unitary $n$-dimensional representations of the fundamental group of $\Sigma$. As we explain below, $M$ is canonically a compact symplectic manyfold, with symplectic form $\w$ {\red\bf (for that the representation in the form $Hom(\pi_1(\Sigma),SU(n))/SU(n)$ is not necessary! We can begin with the bundle interpretation)}.

Let $M$ be a real compact symplectic manifold with the symplectic form $\w$. By a K\"{a}hler polarization of $M$ we mean a complex structure $I:H^0(TM)\to H^0(TM)$, $I^2=-1$ on the tangent vector bundle, satisfying the following compatibility and positivity conditions with the symplectic structure:
\begin{equation}\label{E:comp}
   \w(X,IY)=\w(Y,IX), \quad \w(X,IX)\  \text{is positive definite}\ (X,Y\in H^0(TM)).
\end{equation}
If the de Rham cohomology class $\frac{1}{2\pi}[\w]$ is integral, then $[\w]$ is the curvature form of a connection on a principal $U(1)$-bundle over $M$, whose first Chern class is $\frac{1}{2\pi}[\w]$. When M is simply-connected, this connection is unique up to gauge equivalence.
Instead of considering the principal $U(1)$-bundle, we consider the associated complex line bundle $L$, and an associated connection $\nabla^L$, both depending on a natural number called level, coming from the character of $U(1)$. Hence to every K\"{a}hler polarization of $M$ we can canonically associate the  finite-dimensional space of globally holomorphic sections of $L$. Then we will obtain a vector bundle over the variety of K\"{a}hler polarizations. In \cite{Flat_conn}, N.Hitchin is looking for a projectively flat connection in that vector bundle.

%First, show that a deformation of the K\"{a}hler polarization is given by a symmetric 2-tensor on $M$. Let $I=I_t$ be a path in the variety of K\"{a}hler polarizations. Indeed, the compatibility conditions \refE{comp} imply that $(X,Y)=\w(X,IY)$ is a Riemann metric on $M$. Besides, by differentiation in $t$ of the first compatibility relation we obtain $\w(X,\dot{I}Y)=\w(Y,\dot{I}X)$, i.e. $\w(X,\dot{I}Y)$ is a symmetric bilinear form also, and as such is representable in the form $\w(X,\dot{I}Y)=(X,GY)$ where $G$ is a symmetric matrix. On the other hand side $\w(X,Y)=(X,I^{-1}Y)$, hence $\w(X,\dot{I}Y)=(X,I^{-1}\dot{I}Y)$. It follows that $I^{-1}\dot{I}=G$ {\red\bf (but then $G$ is skew-symmetric)}, or $\dot{I}=IG$. It exactly means that a deformation of the K\"{a}hler polarization $I$ is given by the symmetric tensor $G$. Observe that by $I^2=-1$ we have $I^{-1}\dot{I}=-\dot{I}I^{-1}$, hence $I^{-1}\dot{I}=G$ also implies $\dot{I}=-GI$ which only by change of coordinates differs from \cite[eq. (1.13)]{Flat_conn}.

First, show that a deformation of the K\"{a}hler polarization is given by a symmetric 2-tensor on $M$. Let $I=I_t$ be a path in the variety of K\"{a}hler polarizations. The compatibility conditions \refE{comp} imply that $\w(X,IY)$ is a Riemannian metric on $M$, hence it is representable in the form $\w(X,IY)=(X,\tilde{G}Y)$ where $(X,Y)=\sum_{i}X_iY_i$, $\tilde{G}$ is a positive definite (in particular, symmetric) matrix. We also have $\w(X,Y)=(X,\tilde{G}I^{-1}Y)$.

By differentiation in $t$ of the first compatibility relation we obtain $\w(X,\dot{I}Y)=\w(Y,\dot{I}X)$, i.e. $\w(X,\dot{I}Y)$ is a symmetric bilinear form too, and as such is representable in the form $\w(X,\dot{I}Y)=(X,GY)$ where $G$ is a symmetric matrix also. On the other hand side $\w(X,\dot{I}Y)=(X,\tilde{G}I^{-1}\dot{I}Y)$. It follows that
$\tilde{G}I^{-1}\dot{I}=G$, hence $\dot{I}=(\tilde{G}I^{-1})^{-1}G$.
Since $\tilde{G}I^{-1}$ is the matrix of $\w$ which is assumed to be fixed, $\dot{I}$ is given by the symmetric tensor $G$.

Vice versa, assume we have a multiparametric family of K\"{a}hler polarizations, and the corresponding family of K\"{a}hler manifolds $M_t$, $t=(t_1,\ldots,t_A)$. For every $t$, construct the line bundle $L=L_t$ as above. Consider the vector bundle of the spaces $H^0(L_t)$ over that family. Given a symmetric tensor $G^{ij}$, we can construct the connection in that vector bundle as follows.  First, we locally define the operator $\Delta=\nabla^L_i(G^{ij}\nabla^L_j)$ in the space of holomorphic sections of $L$. As it is proved in \cite{Flat_conn}, for every $a=1,\ldots,A$ there exists locally defined operator $P_a: H^0(L)\to H^0(L)$ with the highest symbol $\Delta$ such that
\begin{equation}\label{E:nabla}
\nabla_a=\partial_a+P_a
\end{equation}
is a globally defined operator $H^0(L)\to H^0(L)$, for all $a=1,\ldots,A$, where $\partial_a=\frac{\partial}{\partial t_a}$.

From now on $M=\mathcal{N}$. We additionally assume the rank and degree of the bundles in $\mathcal{N}$ to be fixed, and, moreover, coprime. Then $\mathcal{N}$ is compact \cite{Nar_Sesh}. These assumptions are not restrictive because $\mathcal{N}$ is just an auxiliary object, while our main goal here is to construct a connection on the moduli space of curves. Observe that $\mathcal{N}$ is canonically symplectic. Indeed, by Narasimhan--Seshadri theorem \cite{Nar_Sesh} $\mathcal{N}\simeq Hom(\pi_1(\Sigma),SU(n))/SU(n)$, and by \cite{AtBot} the last is nothing but the space of flat connections of the form $\partial+\varphi$, where $\varphi$ is a skew-Hermitian matrix valued 1-form on $\Sigma$. Given two such matrix valued 1-forms $\a$, $\b$, now representing tangent vectors to $\mathcal{N}$, we define $\w(\a,\b)=\int\limits_\Sigma \tr(\a\wedge\b)$.

Next, we will specify the family of K\"{a}hler polarizations, and tensor $G^{ij}$ so that $[\nabla_a,\nabla_b]=\l_{ab}\cdot \mathbf{1}$, $\l_{ab}\in\C$.
%òóò àññîöèèðîâàòü êýëåðîâû ïîëÿðèçàöèè ñ êîìïëåêñíûìè ñòðóêòóðàìè íà $\Sigma$
First, choose a Hodge star operator (a complex structure) $*$ on $\Omega^1(\Sigma)$:
\begin{align*}
  & *: \Omega^1(\Sigma)\to\Omega^1(\Sigma), \\
  & *^2=-1.
\end{align*}
Then the corresponding operator on Higgs fields $I:\Phi\to *\Phi$ (obtained by operating entry by entry on the matrix $\Phi$) is a K\"{a}hler polarization on $M$. A deformation of the complex structure on $\Sigma$ is given by a Beltrami differential $\b_a=\phi_a\partial_z\otimes d\overline{z}$. We define the corresponding symmetric 2-tensor on $T^*{\mathcal N}$ by means the relation
\begin{equation}\label{E:symtens}
   G_a(\Phi,\Phi)=\int_{\Sigma}\tr(\Phi^2)\b_a.
\end{equation}
Now we possess all necessary to define the connection by means of \refE{nabla}.
\begin{theorem}[\cite{Flat_conn}]\label{T:pflat}
Connection $\nabla$ defined by means of the relation \refE{nabla} with the symmetric tensor \refE{symtens} is projectively flat.
\end{theorem}
\begin{proof}[Sketch of proof]
Indeed, $[\nabla_a,\nabla_b]=\partial_aP_b-\partial_bP_a+ [P_a,P_b]$. Here, $\partial_aP_b-\partial_bP_a$ is a second order operator, and $[P_a,P_b]$ is a 3rd order operator whose principal symbol is the Poisson bracket of the symbols of $P_a$ and $P_b$ \cite{GuiSt}. Observe that the tensor $G$ determining the principal symbols is a linear combination of 2nd order Hitchin Hamiltonians, hence by \refT{Hitch} the principal symbols are Poisson commuting. We obtain that $[\nabla_a,\nabla_b]$ is a second order operator.

Now consider the exact sequence of sheaves of differential operators
\[
  0\to \mathcal{D}^1(L)\to \mathcal{D}^2(L)\to S^2T\to 0
\]
(saying that the 2nd order differential operator is a sum of its principal symbol and a 1st order operator), and the exact sequence of cohomology groups
\[
\begin{CD}
 0\to H^0(M,\mathcal{D}^1(L))\to H^0(M,\mathcal{D}^2(L)) \to H^0(M,S^2T) @>{\delta}>> H^1(M,\mathcal{D}^1(L)) \to\ldots\quad .
\end{CD}
\]
Without going into details, we state following \cite{Flat_conn} that $\d$ is an injection, and there are no holomorphic vector fields on $M$. The first implies that $H^0(M,\mathcal{D}^1(L))\to H^0(M,\mathcal{D}^2(L))$ is an epimorphism, i.e. any holomorphic 2nd order operator actually is of the 1st order. The second implies that $H^0(M,T)=0$, hence the last operator is of order 0, i.e. a holomorphic function. Since $M$ is compact, such functions are just constants, hence $[\nabla_a,\nabla_b]=const\cdot \mathbf{1}$.
\end{proof}

%%%%%%%%%%%%%%%%%%%%%%%%%%%%%%%%%%%%%%%%%%%%%%%%%%%%%%%%%%%%%
%\subsection{Algebraic-geometric approach and Shottki parametrization}
%%%%%%%%%%%%%%%%%%%%%%%%%%%%%%%%%%%%%%%%%%%%%%%%%%%%%%%%%%%%%
\subsection{Specific approach for rank 2 genus 2 Hitchin systems}\label{SS:Previato}
In \cite{Nar_Ram} M.S.Narasim\-han and S.Ramanan showed that the space of semi-stable rank 2 holomorphic vector bundles on a Riemann surface is canonically isomorphic to $\mathbb{P}\Theta_2$. Here, $\Theta_2$ denotes the space of level 2 $\theta$-functions, i.e.
\[
  e^{2(\pi i n\cdot\tau n + 2\pi i n\cdot u)}\theta(u+m+\tau n)= \theta(u)
\]
(where $m,n\in\Z^2$, $\tau$ is the period matrix), and $\mathbb{P}$ stays for the projectivization.

Based on this result, B.van Geemen and E.Previato \cite{Previato}, B.van Geemen and A.J. de Jong \cite{Geemen_Jong} gave a more specific description of the rank 2, genus 2 Hitchin systems. For genus 2 curves, $\Theta_2$ is 4-dimensional, hence $\mathbb{P}\Theta_2\cong \mathbb{P}^3$, thus they were seeking in \cite{Previato} for an integrable system on $T^*\mathbb{P}^3$. The construction of Hamiltonians in \cite{Previato} is as follows.  Let $q\in\mathbb{P}^3$, $q=(x:y:z:t)$, $p\in(\mathbb{P}^3)^*$. Define $\epsilon_i(q)\in(\mathbb{P}^3)^*$ by means the relations
\begin{align*}
 &\epsilon_1 = (y : -x : t : -z),\ \epsilon_3 = (z : t : -x : -y),\ \epsilon_5 = (t : z : -y: -x) \\
 &\epsilon_2 = (y : -x : -t : z),\ \epsilon_4 = (z : -t : -x : y), \ \epsilon_6 = (t : -z : y : -x).
\end{align*}
Let $x_{ij}$ be Klein coordinates (see \cite{Previato} for the definition) of the line in $(\mathbb{P}^3)^*$ coming through the points $\epsilon_i(q),p$. Then
\begin{equation}\label{E:GPHam}
  H_i(p,q)=\sum_{j\ne i}^6\frac{x_{ij}^2}{z_i-z_j}.
\end{equation}
where $w^2=\prod_{j=1}^{6}(z-z_j)$ is the equation of the base curve. Commutativity of the Hamiltonians \refE{GPHam} with respect to the symplectic structure on $T^*\mathbb{P}^3$, as well as explicit expressions for them, are established in \cite{Previato} by means of computer calculations using ``Mathematica". The identification of \refE{GPHam} as Hitchin Hamiltonians has been done later, in \cite{Geemen_Jong}, in course of examination of the Hitchin connection (see \refSS{class_app}) in the genus 2 rank 2 case. Analytically, commutativity of Hamiltonians has been proved later in \cite{Gaw} (see \refSS{Gawed}).

%%%%%%%%%%%%%%%%%%%%%%%%%%%%%%%%%%%%%%%%%%%%%%%%%%%%%%%%%%%%%
%%%%%%%%%%%%%%%%%%%%%%%%%%%%%%%%%%%%%%%%
\section{Parabolic and parahoric  Hitchin systems}\label{S:Parab}
%In this section, \edit{we will define parabolic vector bundles and parabolic Hitchin systems on a smooth projective curve over a field $k$. We will mainly work on them from a point view of algebraic geometry.}\wb{How do you think that we assume $k$ to be an algebraically closed field? This will help us avoid terminology "geometrically connected" etc.}

In the previous section we have presented constructions of Hitchin systems on compact Riemann surfaces, and of the corresponding flat connections on the moduli spaces of curves without punctures. In the following, we shall define the analogs of those objects for the Riemann surfaces with punctures (at least for the purposes of Conformal Field Theory). The corresponding integrable systems are referred to as parabolic Hitchin systems, for the reason that the corresponding Higgs bundles have a special structure, called parabolic structure, at the punctures of the underlying Riemann surface. A further generalization of parabolic Hitchin systems is referred to as parahoric Hitchin systems. Since our way to treat parabolic Hitchin systems has an algebraic nature, we start working with Hitchin systems over an algebraically closed field $k$. Actually, some results mentioned in the following also work over an arbitrary field (with a mild assumptions in characteristic 2), and such kinds of results are important for the proof of ``topological mirror symmetry".
%\wb{Sorry Oleg, I do not understand the ``flat connection on the moduli of Riemann surfaces". There is a projectively flat connection on the vector bundle of conformal blocks over the moduli spaces, which are sections of generalized theta divisors. But certainly, this is not the flat connection you mentioned here. Moreover, it seems that you are claiming the construction of the flat connections are closely related with Hitchin systems? Can you provide me with some intuitions?}
%%%%%%%%%%%%%%%%%%%%%%%%%%%%%%%%%%%%%%%%%%%%%%%%%%%%%%%%%%%%%
\subsection{Parabolic Vector Bundles}
In this subsection, we first introduce the concept of parabolic vector bundles.
% And our goal is to present the main results of Mehta-Seshadri\cite[MS], where they prove a parabolic analogue of Narashiman-Seshadri's work on the correspondence between stable vector bundles and irreducible unitary local systems on a hyperbolic Riemann surface.

Let $X$ be a smooth projective curve of genus $g$ over the field $k$ which we assume to be algebraically closed. We fix a finite subset $D\subset X$, which we shall also regard as a reduced effective divisor on $X$. We then require that $2g-2+\deg D>0$
(\footnote{$g=0$ is allowed, and then $\deg D\ge 3$.}).
 We also fix a positive integer $r$ which will be the rank of  vector bundles on $X$.

To define a parabolic structure on a vector bundle we need to specify a quasi-parabolic structure and weights for each $x\in D$. Quasi-parabolic structure consists of a finite sequence $m^\cdot (x)=(m^1(x), m^2(x), \dots, m^{\sigma_x}(x))$ with $\sum_{i=1}^{\sigma_x}m^{i}(x)=r$ for each $x\in D$. We denote the quasi-parabolic structure simply by $P$. Weights are given by a choice of a set of rational numbers $0\leq \alpha_1(x)<\cdots <\alpha_{\sigma_x}(x)<1$(\footnote{These weights can be chosen to be real numbers which provides the same coarse moduli space as any choice of rational weights which are in a sufficiently small neighborhood.}), denoted by $\alpha(x)$, for each $x\in D$. And we denote the set of weights $\{\alpha(x)\}_{x\in D}$ by $\alpha$. The set $(X, D, P, \alpha)$ will be referred to as the parabolic type.

\begin{definition}\label{def:parabolic vb}
	A parabolic vector bundle of type $(X, D, P, \alpha)$ is a rank $r$ vector bundle $\mathcal{E}$ on $X$ which for every  $x\in D$ is endowed with a filtration $\mathcal{E}|_{x}=F^0(x)\supset F^1(x)\supset \cdots \supset F^{\sigma_x}(x)=0$ such that $\dim F^{j-1}(x)/ F^{j}(x)=m^j(x)$ and with the weight $\alpha(x)$.
\end{definition}
\begin{remark}
	The definitions of (quasi-)parabolic vector bundles over a Riemann surface were first introduced by Mehta and Seshadri \cite{MS80}. Then they were  generalized  to varieties parametrized by very general schemes by Yokogawa \cite{Yo93C}.
\end{remark}

We now define the \emph{parabolic degree (or $\alpha$-degree)} of $\mathcal{E}$ to be
\[\pdeg(\mathcal{E}):=\deg(E)+\sum_{x\in D}\sum_{j=1}^{\sigma_x}\alpha_{j}(x)m^j(x).\]
%And the \emph{parabolic slope or $\alpha$-slope} of $\mathcal{E}$ is given by \[\pmu(\mathcal{E})=\frac{\pdeg(\mathcal{E})}{r}.\]

\begin{definition}
	A parabolic vector bundle $\mathcal{E}$ is said to be stable(resp. semistable),  if for every subbundle $F\subsetneq E$, we have
	\[
	\frac{\pdeg(\mathcal{F})}{r (\mathcal{F})}<\frac{\pdeg(\mathcal{E})}{r}\ (\text{resp.}\leq),
	\]
	where the parabolic structure on $\mathcal{F}$ is  inherited from $\mathcal{E}$.
\end{definition}
\begin{definition}
	Given a semistable parabolic vector bundle $\mathcal{E}$, there is a filtration, called Harder-Narasimhan filtration:
	\[
	\mathcal{E}=\mathcal{E}_{\ell}\supsetneq\mathcal{E}_{\ell-1}\supsetneq\cdots\supsetneq\mathcal{E}_0=0,
	\]
	 such that for each $1\le i\le\ell$, $\mathcal{E}_{i}/\mathcal{E}_{i-1}$ is stable.  We define $\gr\mathcal{E}=\oplus_{i=1}^{\ell}\mathcal{E}_{i}/\mathcal{E}_{i-1}$. Though Harder-Narasimhan filtration may not be unique, the associated $\gr(\mathcal{E})$ only depends on $\mathcal{E}$. Two semistable parabolic vector bundles $\mathcal{E},\mathcal{E}'$ are said to be equivalent if $\gr\mathcal{E}\cong\gr(\mathcal{E}')$
 \end{definition}
With the help of geometric invariant theory, one can construct the coarse moduli space of semistable parabolic vector bundles of \emph{parabolic degree 0}(\footnote{The coarse moduli spaces can be constructed for general parabolic degrees. Here for simplicity and for describing a correspondence  with unitary local systems, we only deal with parabolic degree 0.}) which we denote as $\mathcal{N}_{P,\alpha}$ (see \cite{Yo93C} where much more general cases have been considered).

In general, $\mathcal{N}_{P,\alpha}$ is a normal projective variety which parametrizes the equivalence classes of semistable parabolic vector bundles of parabolic type $(P,\alpha)$ on $X$. For a generic choice of weights $\alpha$, semistable parabolic bundles are parabolically stable. In this case, the moduli space $\mathcal{N}_{P,\alpha}$ is a smooth projective variety. In the rest of the paper, we always assume $\alpha$ to be generic and we omit $\alpha$ in the index for simplicity.
%%%%%%%%%%%%%%%%%%%
\subsubsection{Relations with Unitary Representations}\label{subsubsec:motivation from unitary rep}In the following, we  explain the motivation of introducing parabolic vector bundles, following Mehta and Seshadri\cite{MS80}.

Let $\mathbb{H}$ be the upper half plane, and $\Gamma$ be a discrete subgroup of $\Aut(\mathbb{H})=\PGL_{2}(\R)$, acting freely on $\mathbb{H}$ and $\mathbb{H}/\Gamma$ has finite volume. We denote by $\mathbb{H}^{+}$, the union of $\mathbb{H}$ and parabolic cusps (\footnote{I.e., points whose stabilizer in $\Gamma$ is cyclic.}). We denote $X=\mathbb{H}^+/\Gamma$, which is a compact Riemann surface, i.e., it is a compactification of $X^{\circ}=\mathbb{H}/\Gamma$. We denote the natural quotient map by $p:\mathbb{H}\rightarrow X^{\circ}$.

Let $\rho:\Sigma\rightarrow \SU(E)$ be a unitary representation of dimension $r$. Then $\mathbb{H}\times E$ admits a natural $\Gamma$-action:
\[
\gamma: (z,e)\mapsto (\gamma z, \sigma(\gamma) e).
\]
The quotient, $\mathbb{H}\times_{\Gamma} E$, can be treated as a vector bundle on $X^{\circ}=\mathbb{H}/\Gamma$ which we denote by $\mathcal{E}^{\circ}$.

For a cusp $x\in X\setminus X^{\circ}$, we denote its stabilizer by $\Gamma_{x}$ which is a cyclic subgroup of $\Gamma$ consisting of unipotent elements. If $x$ is the image of $\infty\in\mathbb{H}^{+}$, for $\delta>0$, we can choose neighborhoods $U_{\delta}(\infty):=\{r+is|s\ge \delta>0 \}$ of $\infty$ which are $\Gamma_x$-invariant. Then $\{x\}\cup \{U_{\delta}/\Gamma_{x}\}$ forms a base of $x$ on $X$.
%	If we take $U_{\epsilon,\delta}(\infty):=\{r+is||r|< \epsilon, s\ge \delta>0 \}$, we can see that for some $\epsilon\ge 0$, $U_{\epsilon,\delta}\rightarrow U_{\delta}(x)$ is surjective.In fact, we can analyze $\mathbb{H}\times_{\Gamma} E$ by looking at $U_{\epsilon,\delta}\times E$.
	 Since the $\Gamma$-action on cusps $\mathbb{H}^+\setminus\mathbb{H}$ might not be transitive, not all cusps $x\in X\setminus X^{\circ}$ are images of $\infty$. However, for any cusp $y\in \mathbb{H}^+\setminus\mathbb{H}$, we can find $h\in\PGL_{2}(\mathbb{R})$ such that $h\infty=y$ which can also be used to identify their neighborhoods. Hence, we can still define such kind of neighborhoods $U_{\delta}(y)$ for all $y\in\mathbb{H}^+\setminus\mathbb{H}$. We denote $U_{\delta}/\Gamma_{x}$ by $U_{\delta}(x)$ and $\{x\}\cup \{U_{\delta}/\Gamma_{x}\}$ by $D_{\delta}(x)$.
\begin{definition}\label{def:sections}
	A holomorphic morphism $F:\mathbb{H}\rightarrow\mathbb{H}\times E$ given by $F(z)=(z,f(z))$ is said to be bounded and  $\Gamma$-equivariant if:
	\begin{itemize}
		\item $\rho(\gamma)f(z)=f(\gamma z)$ for all $z\in\mathbb{H}$,
		\item For all cusps $x$, $f$ is bounded in $U_{\delta}(y)$ for $\delta\gg 0$ where $x$ is the image of $y\in \mathbb{H}^+\setminus \mathbb{H}$.
	\end{itemize}
\end{definition}
%\begin{remark}
%	 \edit{ Notice that we require the boundness for all $\delta\gg 0$, then it is not difficult to see that the boundedness defined here is independent of the choice of $g$.}
%\end{remark}
%\wb{Because that an "equivariant section: $\mathbb{H}\rightarrow \mathbb{H}\times E$" will provide a "section : $\mathbb{H}/\Gamma\rightarrow \mathbb{H}\times_{\Gamma}E$", I think it reasonable to call $F:\mathbb{H}\rightarrow \mathbb{H}\times E$ an "equivariant section" instead of an "equivariant function". And I am sorry that I do not understand your explanation in the email. Why do we have equivariant sections on the quotient if we use" equivariant section" here? }
%\begin{remark}
%	A cusp $x\in X$ can be the image of a cusp $x'\in\mathbb{H}^+$ which may not be $\infty$, i.e., $x'\in \mathbb{R}\subset \mathbb{H}^+$. But we can find a $g\in\PGL_{2}{\mathbb{R}}$ such that $g\infty =x'$ and we can still define the boundedness of $f$ in neighborhoods $g U_{\epsilon,\delta}$ of $x'$.
%\end{remark}

\begin{remark}
	In \cite[Definition 1.1]{MS80}, Mehta and Seshadri call them $\Gamma$-invariant sections in $\mathbb{H}^{+}$.
\end{remark}
It is quite straightforward to see that there is a local version of the boundedness near any cusp. Now let us look at $\Gamma$-equivariant $f=(f_1(z),\ldots, f_r(z))$ near the cusp $x$. Without loss of generality, we may assume that $x$ is the image of $\infty$ and for simplicity, we assume $\Gamma_{x}$ is generated by
$
\gamma_{x}:z\mapsto z+1,
$
and choose a base $\{e_i\}_{i=1}^{r}$ of $E$ such that:
\[
\rho(\gamma_{x})=\begin{pmatrix}
	\exp^{2\pi\sqrt{-1} \alpha_1} & & & \\
	& 	\exp^{2\pi\sqrt{-1} \alpha_2}&&\\
	&& \dots&&\\
	&&&	\exp^{2\pi\sqrt{-1} \alpha_r}
\end{pmatrix}
\]
where $0\le \alpha_1\le\alpha_2\le\ldots \le\alpha_r< 1$. Then $\Gamma$-equivariance implies that:
\[
f_j(z+1)=\exp^{2\pi\sqrt{-1} \alpha_j}f_j(z).
\]
In particular, we may write
\[
f_j(z)=\exp^{2\pi\sqrt{-1} \alpha_j z}g_j(\tau),
\]
here $\tau=\exp^{2\pi\sqrt{-1} z}$ can be treated as a local coordinate at the cusp!
%\wb{Since we assume that $x$ is the image of $\infty$ and the stabilizer of $\infty\in \mathbb{H}^+$ is generated by $z\rightarrow z+1$, the analytic map $\exp^{2\pi\sqrt{-1}\times *}$ maps neighborhoods of $\infty$ to neighborhoods of $x$.}
Hence the boundedness of $F$ amounts to saying that $g$ is holomorphic at $x\in X$.
%\begin{definition}
%	\edit{Near each cusp $x$, the bundle $\mathbb{H}\times_{\Gamma}E$ can be trivialized as $U_{\delta_x}(x)\times E$ for $\delta_x\gg 0$, we can glue  $\{D_{\delta_x}(x)\times E\}_{x\in X\setminus X^{\circ}}$ with the vector bundle $H\times_{\Gamma}E$ over $X^{\circ}$ to obtain a vector bundle over $X$. We denote it by $\mathcal{E}$.}
%\end{definition}
%
Moreover, we can see that bounded sections near $\infty$:
\[
z\mapsto \exp^{2\pi\sqrt{-1}\alpha_{j}z}e_j, \; 1\le j\le r
\]
provide a trivialization of $H\times_{\Gamma}E$ in the neighborhood $U_{\delta}(x)$. Then we can extend the vector bundle $\mathbb{H}\times_{\Gamma} E$ on $X^{\circ}$ to a vector bundle on $X$ denoted by $\mathcal{E}$.

%\edit{We can see that for each $x\in X\setminus X^{\circ}$, sections of $\mathcal{E}$ over the $ D_{\delta}(x)$ are sections over $U_{\delta}$ which are bounded and $\Gamma_x$-invariant.}
%%\wb{Here, I do not think that we should write $\mathcal{E}$ as $\mathbb{H}^+\times_{\Gamma} E$ because the action of $\Gamma$ on $\mathbb{H}^+ $ is no longer free.}
%\wb{I still cannot understand your explanation in the email where we define $\mathcal{H}^+\times_{\Gamma}E$ using equivalence determined by $\Gamma$-orbit. For example, if $x$ is the cusp corresponding to $\infty$, then the fiber over it will be $(\infty\times E)/\Gamma_{x}$ by the definition of the equivalence relation, which is not $E$. Here, I extend the bundle by gluing with trivial bundles locally around each cusp, especially, such a gluing over the disks have just omitted the action of the stabilizer at cusps. }
%
Let $\mathcal{E}_1,\mathcal{E}_2$ be two vector bundles on $X$ coming from two unitary representations $\rho_i:\Gamma\rightarrow \GL_{r_i}(E_i), i=1,2$.

\begin{definition}
	We say a holomorphic morphism
    \begin{align*}
    	\phi:&\mathbb{H}\times E_1\rightarrow \mathbb{H}\times E_2\\
    	&(z,e)\rightarrow (z, \phi(z)e)
    \end{align*}
	is bounded and $\Gamma$-equivariant if:
	\begin{itemize}
		\item $\phi(\gamma z)(\rho_1(\gamma)e)=\rho_2(\gamma)(\phi(z)e)$,
		\item each term of $\phi(z)$ is bounded in the neighborhoods $U_{\delta}(x)$ for all cusps $x$ and all $\delta\gg 0$.
	\end{itemize}
\end{definition}
Now for $E_1, E_2$, we choose bases $\{e_{\ell}\}_{\ell\le r_1},\{d_{m}\}_{m\le r_2}$ such that
\[
\rho_1(\gamma_{x})=\begin{pmatrix}
	\exp^{2\pi\sqrt{-1} a_1} & & \\
	& \dots&\\
	&&	\exp^{2\pi\sqrt{-1} a_{r_1}}
	\end{pmatrix},
	\rho_2(\gamma_{x})=\begin{pmatrix}
		\exp^{2\pi\sqrt{-1} b_1} & &  \\
		& \dots&\\
		&&	\exp^{2\pi\sqrt{-1} b_{r_2}}
	\end{pmatrix}.
\]
We put $\phi(z)e_{\ell}=\sum _{m}\phi_{\ell m}(z)d_m$. Then $\Gamma$-equivariance implies that:
\[
\phi_{\ell m}(z+1)=\exp^{2\pi\sqrt{-1}(b_{m}-a_{\ell})}\phi_{\ell m}(z).
\]
Or in other words:
\[
\phi_{\ell m}(z)=\exp^{2\pi\sqrt{-1}(b_{m}-a_{\ell})z}g_{\ell m}(\tau).
\]
Then the boundedness requires that:
\[
g_{\ell m}(0)=0, \;\text{if}\; b_{m}<a_{\ell}.
\]

We now look at the special case that $\phi$ is an automorphism of $E$. We now rewrite $\{a_{\ell}\}_{\ell\le r}$ as:
\begin{align}\label{eq: multiplicity of eigenvalues}
	&a_{1}=a_{2}=\ldots=a_{m^{1}(x)}\\\nonumber
	&a_{m^1(x)+1}=\ldots =a_{m^{1}(x)+m^{2}(x)}\\\nonumber
	&\cdots\\\nonumber
	&a_{\sum_{j=1}^{\sigma_x-1}m^{j}}=\ldots=a_{r}
\end{align}
Now consider a filtration $F^{\cdot}$ on $\mathcal{E}_{x}$ defined as follows:
\begin{align*}
	&F^{0}\mathcal{E}_x=\mathcal{E}_p\\
	&F^{1}\mathcal{E}_x=\text{subspaces generated  by} e_{m^1(x)+1},e_{m^1(x)+2},\ldots, e_{n}\\
	&F^{2}\mathcal{E}_x=\text{subspaces generated  by} e_{m^1(x)+m^2(x)+1},e_{m^1(x)+m^2(x)+2},\ldots, e_{r}\\
	&\ldots\\
	&F^{\sigma_x-1}\mathcal{E}_x=\text{subspaces generated  by} e_{\sum_{j=1}^{\sigma_x-1}m^{j}(x)+1},e_{\sum_{j=1}^{\sigma_x-1}m^{j}(x)+2},\ldots, e_{r}\\
\end{align*}
And for any bounded $\Gamma$-equivariant morphism $\phi$, we have that:
\[
g(0)(F^{j}\mathcal{E}_x)\subset F^{j}\mathcal{E}_x,
\]
where $g=(g_{\ell m})$. This motivates the definition of filtrations at cusps (marked points) and also the definition of parabolic endomorphisms.

We choose the weights at $x$ to be $0\le \alpha_{1}(x)<\alpha_{2}(x)<\cdots<\alpha_{\sigma_x}(x)<1$ where following \eqref{eq: multiplicity of eigenvalues},
\[\alpha_{1}(x):=a_{1}=a_{2}=\ldots=a_{m^{1}(x)},\] and so on.
%\begin{remark}
%	Here, we emphasize that for the convenience of describing eigenvalues, we use $0\le \alpha_1\le\alpha_2\le\ldots \le\alpha_r< 1$. But to coincide with the weight data, we actually use \eqref{eq: multiplicity of eigenvalues} to rewrite them.
%\end{remark}
To summarize:
\begin{theorem}\cite[Corollary 1.10, Theorem 4.1]{MS80}\label{thm:parabolic correspondence}
	Let $X=\mathbb{H}^{+}/\Gamma$ as before. And for each $x\in X\setminus X^{\circ}$, we choose $\{m^{i}(x)\}$ and weights $0\le \alpha_{1}(x)<\alpha_{2}(x)<\ldots \alpha_{\sigma_x}(x)<1$ as above. Then
	\begin{itemize}
		\item the parabolic degree $\pdeg(\mathcal{E})=0$.
		\item the coarse moduli space $\mathcal{N}_{P,\alpha}$ is homeomorphic to the space of the equivalence classes of unitary representations of $\Gamma$ with the image of the generator of the stabilizer $\Gamma_p$  being conjugate to the diagonal matrix $(\exp^{2\pi\sqrt{-1} a_1,\ldots,\exp^{2\pi\sqrt{-1} a_r}})$ for each $x\in X\setminus X^{\circ}$.
		\item stable bundles correspond to irreducible representations.
	\end{itemize} 	
\end{theorem}

\subsection{Moduli of Parabolic Higgs Bundles}
In this subsection, we first talk about Yokogawa's work on the construction of the coarse moduli spaces of (semi-)stable parabolic pairs on smooth projective varieties.
\begin{definition}\label{def:parabolic Higgs bundle}
	A parabolic Higgs bundle is a pair $(\mathcal{E}, \theta)$ where $\mathcal{E}$ is a parabolic vector bundle as above and $\theta$, called a parabolic Higgs field, is an $\mathcal{O}_{X}$-homomorphism
	$$\theta : \mathcal{E}\to \mathcal{E}\otimes_{\mathcal{O}_X} \omega_{X}(D)$$
	 with the property that it takes each $F^j(x)$ to $F^{j+1}(x)\otimes_{\mathcal{O}_{X}}\omega_{X}(D)|_{x}$. By $\omega_X(D)$, we mean the space of meromorphic differentials having no poles outside D. In particular, a parabolic Higgs field is nilpotent at marked points.
\end{definition}
An endomorphism of the parabolic bundle $\mathcal{E}$ is a vector bundle endomorphism of $E$ which preserves the filtrations $F^\cdot(x)$. We call this a \emph{strongly parabolic endomorphism} if it takes $F^i(x)$ to $F^{i+1}(x)$ for all $x\in D$ and $i$. We denote the subspaces of $\End_{\mathcal{O}_X}(E)$ defined by these properties as
$$
ParEnd(\mathcal{E}) \text{  resp.\  } SParEnd(\mathcal{E}).
$$
Similarly we can define the sheaf of parabolic endomorphisms and that of strongly parabolic endomorphisms, denoted by $\mathcal{P} ar\mathcal{E} nd(\mathcal{E})$ and  $\mathcal{S}\mathcal{P} ar\mathcal{E} nd(\mathcal{E})$ respectively. Hence, a parabolic Higgs field $\theta$ lies in $SParEnd(\mathcal{E})\otimes\omega_{X}(D)$.

\begin{remark}  Unfortunately conventions regarding parabolic Higgs bundles vary in the literature.
	For example, a weakly parabolic Higgs bundle in \cite{SWW22} is referred to as a ``parabolic Higgs
	bundle" in \cite{LM10}, a strongly parabolic Higgs bundle in \cite{SWW22} is referred to as a ``parabolic Higgs bundle"
	in \cite{BR94} and \cite{SS95}. What was worse is that the two kinds of parabolic Higgs bundles have very different moduli spaces and Hitchin systems. cf \cite{Yo93C} and \cite{LM10}. The moduli spaces of the strong ones are symplectic while the moduli of the weak ones are Poisson.
	
\end{remark}	
 \emph{In this survey, for simplicity, a parabolic Higgs bundle is always assumed to be strongly parabolic.} In analogy with the non-parabolic case, we can define (semi-)stability of parabolic Higgs bundles.
\begin{definition}
	A parabolic Higgs bundle $(\mathcal{E},\theta)$ is said to be stable(resp. semistable),  if for every $\theta$-invariant subbundle $F\subsetneq E$, we have
	\[
	\frac{\pdeg(\mathcal{F})}{r(\mathcal{F})}<\frac{\pdeg(\mathcal{E})}{r}\ (\text{resp.}\leq),
	\]
	where the parabolic structure on $\mathcal{F}$ is  inherited from $\mathcal{E}$.
\end{definition}

Similarly to the construction of the coarse moduli space of semistable parabolic vector bundles, one can construct the coarse moduli space of semistable parabolic Higgs bundles which we denote as $\mathcal{M}_{P,\alpha}$ (see \cite{Yo93C} for example). When the weights $\alpha$ are chosen generically, semistable parabolic Higgs bundles are parabolically stable and in this case, the moduli space $\mathcal{M}_{P,\alpha}$ is a smooth quasi-projective variety. In the rest of this paper, we always assume $\alpha$ to be generic and we omit $\alpha$ from the index for simplicity.
%
%\wb{Here, a strictly semistable Higgs bundle is a semistable Higgs bundle, but not stable. There is a terminology called "very stable (parabolic) vector bundles" which refers to stable vector bundles which admits no non-zero nilpotent Higgs fields. Here  I replace "strictly semistable" by the sentence " semistable parabolic Higgs bundles are parabolically stable."}
\subsubsection{Non-abelian Hodge theorem} In Section \ref{subsubsec:motivation from unitary rep}, following Mehta-Seshadri, we explain the motivation of introducing (stable) parabolic vector bundles and state the correspondence between stable parabolic vector bundles and unitary representations of the fundamental group of the puncture Riemann surface, (see Theorem \ref{thm:parabolic correspondence}). All of these results can be seen as a parabolic analogue of Narasimhan-Seshadri's results on unitary local systems and stable bundles on Riemann surfaces. Now we talk about Simpson's non-abelian Hodge theorem for parabolic Higgs bundles \cite{Sim90}, which can be treated as an analogue of Mehta-Seshadri's results in the setting of Higgs bundles. We follow the nice and concise description by Mellit \cite[\S 7.6]{Mellit20}.
%\wb{To my knowledge, it is perhaps Simpson who called this non-abelian Hodge theorem. Local systems are representation of fundamental groups, i.e., $\Hom(\pi_1,\GL_n(\mathbb{C}))$. A Higgs bundles (consisting of a vector bundle and a Higgs field) which can be seen as an element in $H^1(X,\GL_n(\mathcal{O}_X))\oplus H^0(X,\mathfrak{gl}_n\otimes\Omega^{1}_X)$. Then the correspondence, requires vanishing first chern classes, roughly says that $H^1(X,\GL_n(\mathbb{C}))=H^1(X,\GL_n(\mathcal{O}_X))_{0}\oplus H^0(X,\mathfrak{gl}_n\otimes\Omega^{1}_X)$. If $n=1$ and we take the first chern class to be 0, then we have $H^1(X,\mathbb{C})=H^1(X,\mathcal{O}_{X})\oplus H^0(X,\Omega_X^1)$ which is the Hodge decomposition.}
Let $X$ be a Riemann surface, and $D$ is a reduced divisor on $X$. We put $X^{\circ}=X\setminus D$. We start from an irreducible representation:
\[
\rho: \pi_{1}(X^{\circ})\rightarrow\GL_r(\mathbb{C}).
\]
 (Notice that the representation here is no longer unitary. Hence the monodromy need not to be semisimple.) For each $x\in D$, we denote \emph{eigenvalues} of the monodromy $\operatorname{C}_x$ around $x$ by
\[
\{\exp^{-2\pi\sqrt{-1}a_i(x)+4\pi b_i(x)}\}_{1\le i\le r}.
\]
where the multiplicity of $\exp^{-2\pi\sqrt{-1}a_i(x)+4\pi b_i(x)}$ is $m^i(x)$. Be careful, $\operatorname{C}_x$ need not to be diagonal. We here only provide notations of its eigenvalues. Similar as in Section \ref{subsubsec:motivation from unitary rep}, for each $x\in D$, we put the weights (which are used to define stability of weakly parabolic Higgs bundles) to be:
\[
\alpha_1(x):=a_1(x)=a_2(x)=\ldots= a_{m^1(x)}(x),
\]
and so on. Notice that, in the following theorem, \emph{we only require that $\theta: F^{j}(x)\mapsto F^{j}(x)\otimes\omega_{X}(D)|_{x}$, i.e, $(\mathcal{E},\theta)$ is a weakly parabolic Higgs bundle.}
\begin{theorem}
 There is a bijection between
	\begin{itemize}
		\item irreducible representations of $\pi_{1}(X^{\circ})$ with monodormies $\operatorname{C}_x$ as above.
		\item stable weakly parabolic Higgs bundles $(\mathcal{E},\theta)$ satisfying that:
		\begin{enumerate}
			\item parabolic degree $\pdeg(\mathcal{E})=0$
			\item  $\theta_x$ acts on $F^{i-1}(x)/F^{i}(x)$ with eigenvalues $\sqrt{-1}b_j(x)$ and for every j
			the Jordan block corresponding to the
			eigenvalue $\sqrt{-1}b_j(x)$ has the same size
			as the Jordan block of the monodromy $C_x$ corresponding
			to the eigenvalue $\exp^{-2\pi\sqrt{-1}a_j(x)+4\pi b_j(x)}$.
%			 has Jordan normal form of same size as the Jordan normal form of the monodromy $\operatorname{C}_x$ corresponding to the eigenvalue $\exp^{-2\pi\sqrt{-1}a_j(x)+4\pi b_j(x)}$.
			  Here $\sum_{1}^{i-1}m^{\ell}(x)<j\le \sum_{1}^{i}m^{\ell}(x)$. In particular, this implies that $\dim F^{i-1}(x)/F^{i}(x)=m^{i}(x)$.
		\end{enumerate}		
	\end{itemize}
\end{theorem}
%\wb{Thank you for the suggestion, this way to present the correspondence between Jordan blocks is much clearer.}
%\wb{From the theorem, the size of the filtration, $\{F^{i-1}(x)/F^{i}\}$,is determined by the multiplicity of eigenvalues instead of the sizes of Jordan normal forms. That's why we only provide notations for eigenvalues of the monodromy. But as you suggested, I need to point out that the monodromy $\operatorname{C}_x$ may not be diagonalizable. The precise reference of Mellit's paper is given now, and I am sorry that I did not provide correct eigenvalues of $\theta_x$in the previous draft which now have been corrected. }
Now let us return to (strongly) parabolic Higgs bundles (see Definition \ref{def:parabolic Higgs bundle}).
\begin{corollary}
	A stable (strongly) parabolic Higgs bundle $(\mathcal{E},\theta)$ uniquely defines an irreducible representation, whose monodromy $\operatorname{C}_x$ is conjugate to a diagonal matrix $$(\exp^{-2\pi\sqrt{-1}\alpha_{1}(x)},\ldots, \exp^{-2\pi\sqrt{-1}\alpha_{\sigma_x}(x)}),$$ where each $\exp^{-2\pi\sqrt{-1}\alpha_{j}(x)}$ has multiplicity $m^{j}(x)$ for $1\le j\le \sigma_x$ and $x\in D$.
\end{corollary}
\begin{remark}
	In the previous theorem, what we consider are actually weakly parabolic Higgs bundles, and in this corollary, we again focus on our parabolic Higgs bundles, which in fact is strongly parabolic. In particular, $\theta_x$ acts as 0 on $F^{i-1}(x)/F^{i}(x)$. This amounts to saying that the Jordan normal form of $\operatorname{C}_x$ is actually the diagonal matrix $\exp^{-2\pi\sqrt{-1}\alpha_{i}(x)}\cdot Id$ of size $m^{i}(x)$.
\end{remark}
%
%\subsubsection{Parabolic(Parahoric) Principal Higgs Bundles}We should also mention Faltings' work on stable principal bundles and flat connections whose main goal is about general Verlinde formulas which relates conformal blocks and global section of generalized theta functions. Interestingly, the notion of parabolic Higgs bundles also appear. But it is only about the Borel reduction at marked points.
\subsection{Parabolic Hitchin Maps and Fibers}
In this subsection, we define parabolic Hitchin maps and focus on their  generic fibers. These are parabolic analogues of classical ones. The main difference with non-parabolic Hitchin maps is that generic spectral curves are singular and the singularities of generic spectral curves play an important role in the geometry of the parabolic Hitchin systems.

Yokogawa \cite[Page 495]{Yo93C} (see also \cite{BK18} and \cite[Section 3]{SWW22}) defined parabolic Hitchin maps as a restriction of the characteristic polynomial map which we explain now.

First, we can define a projective morphism (we call it as the characteristic polynomial map) from the moduli space (and even moduli stack) of parabolic Higgs bundles to the affine space
$\prod_{i=1}^r \mathbf{H}^0(X,(\omega_X(D))^{\otimes i})$ and pointwisely given by the characteristic polynomial of the parabolic Higgs field ${\theta}$ as
\[
\text{char}_P:\mathcal{M}_P\to \mathbf{H}_D:= \prod_{i=1}^r \mathbf{H}^0(X,(\omega_X(D))^{\otimes i}),\ \ (\mathcal{E},\theta)\mapsto (a_1(\theta),\cdots, a_r(\theta)),
\]
where $a_i=\operatorname{Tr}(\wedge^r\theta)$, in particular, they are the coefficients of the characteristic polynomial.
%\wb{I try to avoid talking about $\lambda$, which will only be used later to define spectral curves.}

By the calculation in \cite{BK18}, the image of $
\text{char}_P$ lies in a subspace $\mathcal{H}_P$ which is defined as follows (see also \cite[Subsection 3.1]{SWW22}.

For each $x\in D$, we reorder $(m^1(x), \cdots, m^{\sigma_x}(x))$ to be $n_1(x)\geq n_2(x)\geq \cdots \geq n_{\sigma_x}$. So $\{n_i(x)\}$ is a partition of $r$ and we use $\mu_j(x)$ to denote the dual partition of it, which means that $\mu_{j}(x)=\#\{\ell:n_{\ell}(x)\geq j, 1\leq\ell\leq \sigma_x\}$. Now we assign a level function $j\rightarrow \gamma_{j}(x), 1\leq j\leq r$, such that $\gamma_{j}=l$ if and only if $$\sum_{t\leq l-1}\mu_{t}(x)< j\leq\sum_{t\leq l}\mu_{t}(x).$$

\begin{proposition}[\cite{SWW22} Theorem 4]
	The image of $\text{char}_P:\mathcal{M}_P\rightarrow \mathbf{H}_D$ lies in the following subspace:
	\[
	\mathcal{H}_{P}:=\prod_{j=1}^r\mathbf{H}^{0}\Big(X,\omega_X^{\otimes j}\otimes\mathcal{O}_X\big(\sum_{x\in D}(j-\gamma_{j}(x))\cdot x\big) \Big)\subset \mathbf{H}_D,
	\]
	which we will call the $\GL_r$-\emph{parabolic Hitchin base}, and we will call $h_P: \mathcal{M}_P\rightarrow \mathcal{H}_P$ the $\GL_r$-\emph{parabolic Hitchin map}.
\end{proposition}
%
%We now introduce spectral curves contained in the projectivization $\mathbb P_X(\mathcal{O}_X\oplus \omega_X(D))$.

We denote the natural projection by $\pi:\operatorname{Tot}(\omega_{X}(D))\rightarrow X$, where $\operatorname{Tot}(\omega_{X}(D))$ is the total space of the line bundle $\omega_{X}(D)$. We denote the tautological section of $\pi^*\omega_{X}(D)$ on $\operatorname{Tot}(\omega_X(D))$ by $\lambda$.
\begin{definition}
	Given a point $a=(a_1,a_2,\ldots,a_r)\in \mathcal{H}_P$, then
	\[
	\lambda^r+\pi^*a_1\lambda^{r-1}+\ldots+\pi^*a_r
	\]
	is a section of $\pi^*(\omega_{X}(D))^{r}$. The spectral curve $X_a$ (associated with $a$) is defined as the zero divisor of the section.
\end{definition}
Notice that the spectral curve $X_a$ is a finite branch covering of the base curve $X$ and we denote the projection by $\pi_a: X_a\rightarrow X$.
%\begin{remark}
%	\edit{In this remark, we give a construction of spectral curves\footnote{But be careful that $a_i$'s here are different from those in the following example where we treat them locally around marked points and for convenience we also mulitply with certain powers of local uniformizers.}. Let $\mathcal{O}(1)$ be the relative tautological line bundle on the projectivization $\mathbb P_X(\mathcal{O}_X\oplus \omega_X(D))$. Then we have $\pi_*\mathcal{O}(1)=\mathcal{O}_{X}\oplus (\omega_{X}(D))^{-1}$, in particular, there is a canonical global section of $\mathcal{O}(1)$, which we denote by $\mu$. Similarly, $\pi_*(\pi^*\omega_{X}(D)\otimes\mathcal{O}(1))\cong\mathcal{O}_X\oplus\omega_{X}(D)$ also has a canonical section, which we denote by $\lambda$. Given a point $a=(a_1,a_2,\ldots,a_r)\in\mathcal{H}_P$,}
%	\[
%	\lambda^r+\pi^*a_1\cdot \lambda^{r-1}\mu+\ldots +\pi^*a_r \mu^{r}
%	\]
%	\edit{gives a section of $(\pi^*\omega_X(D))^{r}\otimes \mathcal{O}(r)$ which then defines a divisor on $\mathbb P_X(\mathcal{O}_X\oplus \omega_X(D))$ and we denote it by $X_a$.\emph{We can observe that, $\mu$ is the constant function $1$. Hence the above section actually defines a divisor in $\omega_{X}(D)$ with defining equation:}}
%	\[
%	\lambda^r+\pi^*a_1\cdot \lambda^{r-1}+\ldots +\pi^*a_r.
%	\]	
%\end{remark}
%
%\wb{And I agree with your opinion that this equation should be related with $R(x,y)$ you wrote down in the comment.}
Defined in this way, the spectral curves (c.f. \cite[$\S$ 3.2 and $\S$ 4.2]{SWW22}) would never be smooth unless all the parabolic structures are given by full flags.
\begin{example}\label{ex: singular spectral curves}
	Let us now consider a simple example. Take the rank $r=4$, and choose a local coordinate $t$ around a marked point $x\in D$. Then the local section $\frac{dt}{t}$ can be treated as a generator (\footnote{The line bundle $\omega_{X}(D)$ locally is a free coherent module of rank 1, hence we can talk about its generator.})
	of the line bundle $\omega_{X}(D)$ around $x$. We take the filtration as:
	\[
	F^0(x)\supset F^1(x)\supset F^2(x)=0,
	\]
	where $\dim F^1(x)=2$. Then we may write:
	\[
	\theta=\Theta \frac{dt}{t},
	\]
	where $\Theta$ is a matrix and each entry is a Laurent power series:
	\[
	\Theta=\begin{pmatrix}
		t*&t*&*&*\\
		t*&t*&*&*\\
		t*&t*&t*&t*\\
		t*&t*&t*&t*\\
	\end{pmatrix},
	\]
	where each star represents an element in $k[[t]]$. If we denote the characteristic polynomial of $\Theta$ as
	\[
	f(\lambda):=\det(\lambda\cdot \text{Id}-\Theta)=\lambda^{4}+a_1(t)\lambda^3+a_2(t)\lambda^2+a_3(t)\lambda+ a_4(t),
	\]
	then $\ord_t a_1(t)=\ord_ta_2(t)=1, \ord_ta_3(t)=\ord_ta_4(t)=2$. In particular, this implies that the spectral curve, locally defined by 	$\det(\lambda Id-\Theta)$, has a singularity at $(\lambda,t)=(0,0)$.
\end{example}
\begin{remark}
Since spectral curves are contained in the surface $\operatorname{Tot}(\omega_{X}(D))$, they have planar singularities. To analyze their singularities, we trivialize the line bundle $\omega_{X}(D)$ to obtain local equations explicitly.
%		Notice that, we \emph{do not} divide the Higgs fields by a global meromorphic differential form with simple poles at $D$ though we can do this. The analysis of singularities is local, and we simply choose a local coordinate to carry out the calculation.
\end{remark}
In the case that generic spectral curves are singular, there is no obvious way to endow rank $1$ torsion-free sheaves on the (singular) spectral curve with a prescribed parabolic structure which fits well with the Higgs fields. In \cite[Theorem 6]{SWW22}, via some non-obvious commutative algebra arguments,  a BNR type correspondence is established in the parabolic case:

\begin{theorem}[parabolic BNR correspondence, \cite{SWW22} Theorem 6]\label{parabolic BNR}
	Assume that the base field $k$ is algebraically closed. For a generic $a\in \mathcal{H}_P$, if $X_a$ is integral, then there is a one-to-one correspondence between:
	\begin{itemize}
		\item[(1)] Parabolic Higgs bundles in the fiber $h_P^{-1}(a)$;
		\item[(2)] Line bundles on the normalization $\tilde{X}_a$ of the spectral curve $X_a$.
	\end{itemize}
\end{theorem}
The theorem was first proved over algebraically closed fields (with a mild assumptions in characteristic 2) in \cite{SWW22} and later was generalized to arbitrary fields in \cite[Theorem 2.3.1]{SWW22t}. First, we can define the following local data:
\begin{definition}\label{local data} A \emph{local parabolic Higgs bundle} is a triple $(V,F^\bullet,\theta)$ satisfying the following conditions:
\begin{itemize}
	\item[(a)] $V$ is a free $\mathcal{O}=k[[t]]$-module of rank $r$, endowed with a filtration $F^\bullet V$: $$V=V^{0}\supset V^1\supset\cdots\supset V^{\sigma}=t\cdot V,$$
	with $\dim\frac{V^{i}}{V^{i+1}}=m_{i+1}$. As before we rearrange $(m_{i})$ as $(n_{i})$ to give a partition of $r$.
	\item[(b)] $\theta:V\rightarrow V$ is a $k[[t]]$-linear morphism and $\theta(V^{i})\subset V^{i+1}$.
\end{itemize}
We call $(V,F^\bullet,\theta)$ a \emph{distinguished local parabolic Higgs bundle} if in addition, the following holds:
\begin{itemize}
	\item [(c)] Let $f(\lambda,t)$ be the characteristic polynomial of $\theta$, then it has a factorization, $$f(\lambda,t)=\prod_{i=1}^{n_{1}}f_{i},$$
	such that each $f_{i}$ is an Eisenstein polynomial (\footnote{Here an Eisenstein polynomial is a polynomial  such that the valuation of its constant term is 1.}) with $\deg(f_{i})=\mu_{i}$. Here $(\mu_1,\ldots,\mu_{n_1})$ as before is the conjugate partition of the partition $(n_i)$. Besides, if $\deg f_{i}=\deg f_{j}$, then the difference of their constant terms lies in $t\cdot k[[t]]\setminus t^{2}\cdot k[[t]]$.
\end{itemize}
\end{definition}
Then comes the resolution of the singularities of the spectral curve $X_a$. Through successive blow-ups, for generic spectral curves, the ramification of the normalized spectral curve $\bar{X}_a$ over a marked point $x\in D$ is given by the conjugate partition of $\{n_i(x)\}$. In particular, for generic characteristic polynomials of parabolic Higgs bundles, the local defining equations of spectral curves indeed provide the expected decomposition in $(c)$ of Definition \ref{local data}. Hence locally around each marked point, generic parabolic Higgs bundles $(\mathcal{E},\theta)$ are \emph{distinguished local parabolic Higgs bundle} ( see Definition \ref{local data}).
\begin{remark}
	The last statement in (c) of Definition \ref{local data} has a geometric meaning. During the successive blow up, it says that the intersection of the strict transform with the exceptional divisor is as simple as possible. This corresponds to the non-degenerate condition in the toric resolution. See \cite[Section 3.6]{Oka10}.
\end{remark}
We write $A:=\mathcal{O}[\lambda]/(f)$, and $A_i:=k[[t]][\lambda]/(f_i(t,\lambda))$. Since each $f_{i}$ is an Eisenstein polynomial, each $A_{i}$ is a DVR (discrete valuation ring) and we put $\tilde{A}:=\prod_{i=1}^{n_{1}}A_i$. Then we have a natural injection $A\hookrightarrow \tilde{A}$ and $\tilde{A}$ can be treated as the normalization of $A$.

\begin{theorem}[Local Parabolic BNR Correspondence]
A distinguished local strongly parabolic Higgs bundle $(V,F^\bullet,\theta)$ induces a principal $\tilde{A}$-module structure on $V$.
\end{theorem}
This theorem implies that, parabolic Higgs bundles with generic characteristics polynomials are actually line bundles over normalized spectral curves. Moreover, given a line bundle $\mathcal{L}$ on a normalized spectral curve $\tilde{X}_a$, there is a natural filtration on $\tilde{\pi}_{a,*}\mathcal{L}$ which coincides with the parabolic structure we expect.
\begin{example}
	Now let us continue the analysis of Example \ref{ex: singular spectral curves}. The filtration corresponds to the Young diagram:
	\[
	\yng(2,2)
	\]
	Obviously the conjugate partition gives the same Young diagram
	. We now decompose the characteristic polynomial as:
	\[
	f(\lambda)=f_1f_2,
	\]
	moreover
	\[
	\deg f_1=\deg f_2=2.
	\]
	There are two points $x_1,x_2$ on the normalization $\tilde{X}_a$ over the marked point $x\in X$. Now consider a line bundle $\mathcal{L}$ over $\tilde{X}_a$, it is endowed with a filtration:
	\[
	\mathcal{L}\supset \mathfrak{m}_{x_1}\mathcal{L}\oplus \mathfrak{m}_{x_2}\mathcal{L}\supset  \mathfrak{m}^{2}_{x_1}\mathcal{L}\oplus \mathfrak{m}^{2}_{x_2}\mathcal{L}=\tilde{\pi}^{*}_{a}\mathfrak{m}_{x}\cdot\mathcal{L}
	\]
	Here $\mathfrak{m}$ denotes maximal ideals. Hence on the push forward $\tilde{\pi}_{a*}\mathcal{L}$,  we have a filtration satisfying that $\dim F^1(x)=2$.
\end{example}
\begin{remark}
	As can be seen from the example, the ramification of normalized spectral curves fits perfectly with the filtration at marked points. Such a condition also holds for generic parabolic symplectic Higgs bundles. But for parabolic $\SO(2r+1)$-Higgs bundles, this \emph{does not hold} any more.
\end{remark}

\subsection{Mirror Symmetry for Hitchin Systems}
%In this subsection, we introduce the conjecture of Hausel and Thaddeus \cite{HT03} on topological mirror symmetry between $\SL_n$-Hitchin systems and $\PGL_n$-Hitchin systems. We can observe that the ``duality" of generic fibers plays an important role in the proof of this conjecture by Groechenig, Wyss and Ziegler \cite{GWZ20}. Then by the Theorem \ref{parabolic BNR}, we also have the topological mirror symmetry in the parabolic setting. And in the end of the subsection, we will talk about the works \cite{DG02},\cite{DP12}, \cite{CZ17} on ``duality" of Hitchin systems for general reductive groups and their Langlands dual.
In this subsection, we will talk about mirror symmetry centering around various Hitchin systems.
\subsubsection{Geometric Langlands and Mirror Symmetry}The geometric Langlands program is a long-lasting and important subject of mathematics which has attracted re-examination from a physical perspective during the last twenty years, for example, by Gukov, Kapustin, and Witten \cite{KW07, GW08, GW10}. Such kind of physical perspective implies that we may treat Geometric Langlands as a mirror conjecture of two moduli spaces of stable Higgs bundles $ \mathcal{M}_{G}(X)$ and $\mathcal{M}_{{}^LG}(X)$, where $X$ is a Riemann surface. Here reductive complex algebraic groups $G, {}^LG$ are related by Langlands duality. As explained in \cite{DP08}, in a certain semi-classical limit, Kontsevich's \cite{Ko95} homological mirror symmetry conjecture \cite{KW07} should take the following form:
\begin{align*}
	\mathcal{D}_{\mathrm{Coh}}^b(  \mathcal{M}_{G}(X)) \sim  \mathcal{D}_{\mathrm{Coh}}^b(  \mathcal{M}_{{}^LG}(X)),
\end{align*}
i.e., an equivalence of the derived categories of coherent sheaves on moduli spaces of Higgs bundles for Langlands dual groups. In fact, if we consider $\mathcal{M}_{G}$ (resp. $\mathcal{M}_{^{L}G}$) as cotangent bundles of $\mathcal{N}_{G}$(resp. $\mathcal{N}_{^{L}G}$) the moduli of $G$ (resp. $^{L}G$) bundles. The derived category of coherent sheaves on $\mathcal{M}_{G},\mathcal{M}_{^{L}G}$ can be treated as a limit of the derived category of $D$-modules on $\mathcal{N}_{G},\mathcal{N}_{^{L}G}$. Roughly speaking, we can deform the ring of differential operators on a manifold $X$ to the ring $\operatorname{Sym^{\bullet}}T_{X}$, which also can be treated as the push forward of $\mathcal{O}_{T^*X}$ to $X$.

In Donagi and Pantev\cite{DP12}, they provide an ``Abelian" version of the above equivalence over an open subset of the Hitchin base. They first show the following:
\begin{theorem}\label{T:DonPan}
	The two Hitchin systems $ \mathcal{M}_{G}(X)$ and  $\mathcal{M}_{{}^LG}(X)$ enjoy the Strominger-Yau-Zaslow (SYZ) mirror symmetry, i.e., generic fibers are torsors over dual Abelian varieties.
\end{theorem}
\begin{remark}\label{rmk:SYZ missing gerbes}
	There is a more delicate condition on the torsor structure which is reflected via certain unitary gerbes. See \cite[\S1,\S3]{HT03} or \cite[\S 6]{GWZ20} for an explanation in the setting of Hitchin systems.
\end{remark}
Then for an open subvariety of the Hitchin moduli spaces, they prove a Fourier-Mukai type equivalence of derived categories:
\begin{theorem}\cite[\S5.3]{DP12}
	Over an open subset of the Hitchin base, over which the fibers are torsors over Abelian varieties, there is an equivalence of derived categories via Fourier-Mukai transform:
	\[
	\mathcal{D}_{\mathrm{Coh}}^b(  \mathcal{M}_{G}(X)) \sim  \mathcal{D}_{\mathrm{Coh}}^b(  \mathcal{M}_{{}^LG}(X)).
	\]
\end{theorem}
\begin{remark}
	Here for simplicity, we omit the discussion of non-neutral component of the Hitchin moduli spaces and the unitary gerbes. Both of these discussions are necessary for the verification of SYZ mirror symmetry in general cases.
\end{remark}
\subsubsection{Topological Mirror Symmetry}\label{SS:TopMirSym}
We now introduce another interpretation of the mirror symmetry from a topological aspect. By Physicists, from two versions of string theory, \textcolor{blue}{type IIA} and \textcolor{purple}{type IIB}, one gets two Calabi-Yau manifolds \textcolor{blue}{$X_A$} and \textcolor{purple}{$Y_B$}.
\begin{itemize}
	\item The physical phenomena are presented by: \\				
	\begin{tabular}{rcl}
		\textbf{symplectic geom.} of \textcolor{blue}{$X_A$} & and & \textbf{complex geom.} of \textcolor{purple}{$Y_B$}.
	\end{tabular}
	\item Hence there exists mirror symmetry:\\
	\begin{tabular}{rcl}
		\textbf{symplectic geom.} of \textcolor{blue}{$X_A$} & $\leftrightsquigarrow$ & \textbf{complex geom.} of \textcolor{purple}{$Y_B$}.
	\end{tabular}
	\item We call a pair of Calabi-Yau manifolds $(X,Y)$ satisfying the (vague) symplectic-complex transformation as a \textbf{mirror pair}
\end{itemize}

One way to describe the abstract and ambiguous ``mirror pair" is to compare their topological invariants, for example, Hodge numbers. To be more precise, given a smooth ``mirror pair" of compact Calabi-Yau manifolds $(M, M^\vee)$ of same dimension, we expect
$$
h^{p,q}(M) = h^{\mathrm{dim M} -p, q}(M^\vee).
$$
If moreover, $M$ (also $M^{\vee}$) is a compact hyperK\"ahler manifold, due to the symmetry in the Hodge diamond, we expect that $h^{p,q}(M) = h^{p, q}(M^\vee)$.

However, in many cases, the mirror counterpart might not be a manifold, but be an orbifold. And we need to use the so-called stringy Hodge numbers which we will define in the following(\footnote{But for simplicity, we only define it for orbifolds which are quotients of manifolds by finite groups.}).
\subsubsection*{Intermezzo: Stringy Hodge Numbers}
Let $\Gamma$ be a finite group acting generically fixed-point freely on a smooth quasi-projective complex variety  $V$ of dimension $n$. For each element $\gamma\in\Gamma$, we put $C(\gamma)$ its centralizer in $\Gamma$. We denote by $[V/\Gamma]$ the quotient stack (orbifold) over $\mathbb{C}$. In the following, we fix a system of primitive roots of units,  $\{\xi_{m}\}_{m\in\mathbb{Z}_{>0}}$. Here for each positive integer $m$, $\xi_m$ is an $m$-th primitive root of unit and for all $m,n>0$, we have:
\[
\xi_{mn}^{n}=\xi_{m}.
\]
See \cite[Section 1: Conventions]{GWZ20}.
\begin{definition}\label{def:stringy polynomials}
	
	The stringy $E$ polynomial of $[V/\Gamma]$ is defined as:
		\begin{equation}\label{eq:stringy E}
		E_{\text{st}}([V/\Gamma];u,v)=\sum_{\gamma\in \Gamma/{\text{conj}}}(\sum_{Z\in\pi_{0}([V^{\gamma}/C(\gamma)])}E(Z;u,v)(uv)^{F(\gamma,Z)}),
		\end{equation}
		where $\Gamma/{\text{conj}}$ is the set of conjugacy classes of $\Gamma$, the second summation is over connected components of $[V^{\gamma}/C(\gamma)]$. Here $F(\gamma,Z)$ is the Fermionic shift defined as follows, let $x\in V^{\gamma}$ with image in $Z$, then $\gamma$ acts on $T_xV$ with eigenvalue $\{\xi^{c_i}\}_{i=1,\ldots,n}$ where $0\leq c_i<\#\langle\gamma\rangle$, $\#\langle\gamma\rangle$ is the order of $\gamma$,  for $i=1,\ldots,n$. Then
		\[F(\gamma,x)=\sum \frac{c_i}{\#\langle\gamma\rangle}.\]
		It is not difficult to see that $F(\gamma,x)$ is locally constant for $x\in V^{\gamma}$, hence we put it as $F(\gamma,Z)$.
		
		We write $Z=[W/C(\gamma)]$, where $W$ is the preimage of $Z$ in $V^{\gamma}$, then:
		\[
		E(Z;u,v)=\sum_{p,q,k}(-1)^{k}\dim (\text{gr}^{W}_{p,q}H_{c}^{k}(W)^{C(\gamma)})u^pv^q.
		\]
\end{definition}
\begin{remark}
	In order to fit into SYZ mirror symmetry, the topological mirror symmetry is an equality between twisted stringy Hodge numbers (or equivalently, twisted stringy $E$-polynomials). Here as in Remark \ref{rmk:SYZ missing gerbes}, we omit the discussion of how the gerbes have to be used to give twisted stringy $E$-polynomials. For a more complete definition, see \cite[\S 2.4]{GWZ20}.
\end{remark}
Now we present the relevant results for Hitchin systems. Let $G=\rm SL_n$, then ${}^L G=\rm PGL_n$, one expects the equality of the (stringy) Hodge numbers.
\begin{theorem}[Conjecture \cite{HT03}, Theorem \cite{GWZ20, MS21}] \label{conjecture}
	Assume $d = \deg L$ and $ d' = \deg L' $ are coprime to $n$, then
	$$
	E_{\text{st}}( \mathcal{M}_{\SL_n, L}) = E_{\text{st}}( \mathcal{M}_{\PGL_n, L'}, \alpha).
	$$

\end{theorem}
	Here $ \mathcal{M}_{\mathrm{SL_n}, L} $ denotes the moduli space of stable $\SL_n$-Higgs bundles $(E, \theta)$ over $\Sigma$, where $E$ is a vector bundle of rank $n$ together with an isomorphism $\mathrm{det}(E) \cong L$, and $\theta \in \mathrm{H}^0(C, \mathrm{End}(E) \otimes K_\Sigma )$ is trace-free. Since $d$ is coprime to $n$, then $ \mathcal{M}_{\mathrm{SL_n}, L} $ is smooth. It is known that $ \mathcal{M}_{\mathrm{PGL_n}, L'} =  \mathcal{M}_{\mathrm{SL_n}, L'} / \Gamma $, where $\Gamma = \mathrm{Jac}(\Sigma)[n]$, the subgroup of $n$-th torsion points. There is natural $\mu_n$ gerbe (\footnote{Here $\mu_n$ is the group of $n$-th roots of unit, which can also be treated as the center of $\mathrm{SL}_n$.}) whose local sections are liftings of the universal $\PGL_n$ bundles to a $\SL_n$ bundle on $\mathcal{M}_{\mathrm{SL_n}, L'} $. By the natural embedding $\mu_n\hookrightarrow \mathrm{U}(1)$, it induces a $\mathrm{U}(1)$ gerbe on $\mathcal{M}_{\mathrm{SL_n}, L'}$ which is $\Gamma$-equivariant and defines the $\mathrm{U}(1)$ gerbe $\alpha$ on $ \mathcal{M}_{\mathrm{PGL_n}, L'}$. And $E_{\text{st}}( \mathcal{M}_{\PGL_n, L'}, \alpha)$ is the twisted stringy $E$-polynomial.

The above theorem is also called the \emph{topological mirror symmetry}, which was first proposed by Hausel-Thaddeus \cite{HT03} and proved by them when $n=2$. It was proved for arbitrary $n$ by Groechenig-Wyss-Ziegler \cite{GWZ20}, and Maulik-Shen \cite{MS21} independently.

Groechenig-Wyss-Ziegler\cite{GWZ20} first show that the equality between twisted stringy $E$ polynomials amounts to an equality between certain $p$-adic integrations. They introduced the notion of \emph{a dual pair of abstract Hitchin systems \cite[Definition 6.9]{GWZ20}}. In particular, they reveal an important relation between the $\text{U}(1)$-gerbes (which actually can be induced from $\mu_n$-gerbes) and the Tate duality for Abelian varieties over $p$-adic fields.

%\begin{theorem}
%	They enjoy the Strominger-Yau-Zaslow (SYZ) mirror symmetry, i.e., generic fibers are torsors over dual Abelian varieties.
%\end{theorem}
%\begin{remark}
%	There is a more delicate condition on the torsor structure which is relected via the unitary gerbe we choose.
%\end{remark}
 Maulik-Shen \cite{MS21} used a sheaf-theoretic approach which relate the cohomology of orbifold locus (i.e., those $Z$'s in the Equation \eqref{eq:stringy E}.) with the cohomology of moduli of lower rank Higgs bundles. They show the isomorphisms between graded pieces of perverse filtrations which is actually a refinement of the equality between twisted stringy Hodge numbers.  Such kind of philosophy finally leads to the celebrated proof of the $P=W$ conjecture  \cite{CMS1, CMS2, MS22} \footnote{T. Hausel, A. Mellit, A. Minets, and O. Schiffmann \cite{HMMS} give a different proof of the $P=W$ conjecture.}.

 \subsubsection*{Topological Mirror Symmetry for Parabolic Hitchin Systems} Now we want to close this subsection with the topological mirror symmetry between parabolic $\SL/\PGL$ Hitchin systems.

\emph{ We now work over an arbitrary field $k$ instead of an algebraically closed field.}
 We first define a numerical invariant from the parabolic type $(D,P)$, which will play an important role in the formulation of topological mirror symmetry for parabolic Hitchin systems.
 \begin{definition}
 	We denote $\Delta_{P}:=\gcd\big \{\#\{\ell|\mu_{\ell}(x)=i\}\big\}_{i=1,\ldots,r;\  x\in D}$, where $\gcd$ stands for the greatest commmon divisor.
 \end{definition}

Analyzing the resolution of singular spectral curves over arbitrary fields, we can show the existence of rational points:
 \begin{proposition}\label{prop:k-rational point}
 	There is a $k$ rational point on $\text{Pic}^{\Delta_P}(\tilde{X}_a)$, here $\text{Pic}^{\Delta_{P}}$ is the  degree $\Delta_{P}$ connected component of the Picard variety.
 \end{proposition}
 \begin{remark}
 	As mentioned before, Theorem \ref{parabolic BNR} holds almost over an arbitrary field, especially for the function field of the Hitchin base. Hence, this provides the information of rational points on generic Hitchin fibers which is important for the calculation of p-adic integrations.
 \end{remark}
 Now we can state the topological mirror symmetry.
\begin{theorem}\label{main theorem}\cite[Theorem 4.3.1]{SWW22t}
    The following equality\footnote{We emphasize that the choice of weights is generic, hence the moduli spaces here are smooth. Smoothness is required in the p-adic integration.}, often referred to as topological mirror symmetry, holds:
 	\begin{equation}
 		E_{\text{st}}^{}(\mathcal{M}_{\SL_n,\mathcal{L}};u,v; e\hat{\alpha})=E_{\text{st}}(\mathcal{M}_{\PGL_n,\mathcal{L}'};u,v;d\check{\alpha}),
 	\end{equation}
 	where $d=\deg\mathcal{L},e=\deg\mathcal{L}'$. And we require that there exists a positive integer $\lambda$ such that $e\equiv \lambda d\ (\emph{mod}\  \Delta_P)$.
\end{theorem}
%\wb{Thank you for the suggestion.This presents the results much better.}
%\wb{Such a number may not exist, since in parabolic cases, we no longer assume that $d,e$ are coprime with $n$. For example, if $\Delta_P=15,e=6,d=5$, then we cannot find such a $\lambda$.}
Here $\hat{\alpha}$ and $\check{\alpha}$ are particular choices of $\mu_n$-gerbes on moduli spaces which satisfy the SYZ mirror symmetry.
\begin{remark}
  Notice that the codimension 2 condition in the definition of dual pair of abstract Hitchin systems \cite[Definition 6.9]{GWZ20} does not hold here because Theorem \ref{parabolic BNR} holds over an open subset, whose complement is of codimension 1. However, this can be compensated by the existence of symplectic forms on $\mathcal{M}_{\SL_n,\mathcal{L}}$. It is even possible to replace the canonical line bundle by other sufficiently ample line bundles over the curve, see \cite{Shen23}.
\end{remark}
In particular, if at one marked point, the parabolic subgroup $P$ is in fact a Borel subgroup, then $\Delta_{P}=1$, hence we can choose $d,e$ arbitrarily. We can give a generalization of \cite[Theorem 3.13]{Gothen19} which proves the topological mirror symmetry in rank 2,3 cases and all parabolic subgroups are Borel subgroups.
\begin{theorem}\label{thm:delta_P is one}
	If $\Delta_P=1$, then for all integer $d,e$, we have:
	\[
	E_{\text{st}}^{}(\mathcal{M}_{\SL_n,\mathcal{L}};u,v; e\hat{\alpha})=E_{\text{st}}(\mathcal{M}_{\PGL_n,\mathcal{L}'};u,v;d\check{\alpha})=E_{\text{st}}^{}(\mathcal{M}_{\SL_n,\mathcal{L}};u,v)=E_{\text{st}}(\mathcal{M}_{\PGL_n,\mathcal{L}'};u,v).
	\]
	In particular, $\Delta_P=1$ holds if there exists a marked $x$ such that $P_x$ is Borel, i.e., a full flag filtration at $x$.
\end{theorem}
In fact, Proposition \ref{prop:k-rational point} shows that there are always $k$-rational points for every generic fibers over an arbitrary field, hence also for $p$-adic fields. Then by the $p$-adic integrations, the twists given by gerbes are actually integrations of the constant function 1, hence there is no change of $E$-polynomials after removing the gerbes.

%\subsection{On Perverse Filtration of Hitchin Maps}
%\emph{This is stated for non-parabolic cases. I do not know whether I should mention it here.}
%This includes various support theorem originated from Ngo's work. And we will also mention Yun's work on global Springer theory who shows that there is an affine Weyl group action on moduli of "parabolic Higgs bundles"(\footnote{The moduli of parabolic Higgs bundles here is not exactly the same as we discussed above.}).
%
%We here should mention a recent progress on the mysterious and interesting "P=W" conjecture. Here "P" means the perverse filtration on the cohomology of the moduli of stable Higgs bundles, "W" means the weight filtration on the cohomology of character varieties. The result of Narashiman-Seshadri presents a diffeomorphism between the moduli of stable Higgs bundles and character varieties. It is quite surprising that such a diffeomorphism (which is definitely not algebraic) exchanges the perverse filtration and weight filtrations.
%
%
%It was introduced by de Cataldo-Hausel-Migliorioni, and proved in rank 2 cases. Very recently, the general case was proved by Shen-Maulik.
%\subsection{Parahoric Higgs Bundles}
%(If needed,)In this subsection, we shall talk about parahoric Higgs bundles on curves.

%%%%%%%%%%%%%%%%%%%%%%%%%%%%%%%%%%%%%%%%
\section{Methods of exact solution}\label{S:Exact}
As a non-linear completely integrable system, a Hitchin system admits two methods of exact solution, the inverse spectral method, and classical method of separation of variables. The inverse spectral method for Hitchin systems  with the structure group $GL(n)$ is due to I.Krichever \cite{Kr_Lax}. We present it in \refSS{ISM} below. In the course of further attempts by one of the authors there were found out certain obstructions to generalization of the Krichevers approach to Hitchin systems with simple classical structure groups (\refSS{ISM_obstr}).

For Hitchin systems, the method of separation of variables goes back to Hurtubise \cite{Hur}. In \cite{GNR} the geometry of separation of variables for Hitchin systems is  discussed in more detail. Finally, in \cite{Sh_FAN_2019} we have shown that focusing on a certain class of base curves (the hyperelliptic curves in our case) enables us to give a constructive definition of Hitchin systems (including the case of simple structure groups), and find out Darboux and action-angle coordinates in terms of separated variables. Here we also give a method of finding the $\theta$-functional solutions for the $GL(n)$ systems, and demonstrate the obstruction for that in the case of simple structure groups.

We begin this section with the Lax representation, due to K.Gawedzki and P.Tran-Ngoc-Bich \cite{Gaw}, for the $SL_2$ Hitchin system as it is given in \cite{Previato} (see \refSS{Previato} of the present paper). It is a historically first Lax representation for a hierarchy of Hitchin flows.
%%%%%%%%%%%%%%%%%%%%%%%%%%%%%%%%%%%%%%%%%%%%%%%%%%%%%
\subsection{Lax matrix approach for $SL_2$ Hitchin systems on genus 2 curves}\label{SS:Gawed}

In \cite{Gaw}, K.Gawedzki and P.Tran-Ngoc-Bich proved that the matrices $(x_{ij})$ \refE{GPHam} (skew-symmetric by construction) satisfy the following $\so_6$ commutation relations
\[
   \{x_{nm},x_{mp}\}=-x_{np}\ \text{for}\ n,m,p\ \text{different}, \{x_{nm},x_{pq}\}=0\ \text{for}\ n,m,p,q\ \text{different},
\]
which directly imply commutativity of the Hamiltonians \refE{GPHam}. They invented the following Lax matrix with the spectral parameter $\zeta$:
\[
   L_{nm}(\zeta)=\zeta x_{nm}+z_n\d_{nm},
\]
and proved that the quantities $\tr L(\zeta)^l$ are generating functions for the system of Hamiltonians equivalent to \refE{GPHam}. Let $\partial_l(\zeta)$ be the generating function for the corresponding time derivatives, i.e. $\partial_l(\zeta)(\cdot)=\{ \tr L(\zeta)^l, \cdot  \}$.
\begin{proposition}[\cite{Gaw}]
Hierarchy of commuting flows given by the Hamiltonians \refE{GPHam} is equivalent to the hierarchy given by the Lax equations $\partial_l(\zeta)L(\zeta')=[M_l(\zeta,\zeta'),L(\zeta')]$ where
\[
  M_l(\zeta,\zeta')=l\frac{\zeta\zeta'}{\zeta-\zeta'}L(\zeta)^{l-1} + l\frac{\zeta\zeta'}{\zeta+\zeta'}L(-\zeta)^{l-1}.
\]
\end{proposition}
It is the main result of \cite{Gaw} that there have been found the action-angle coordinates for the $SL_2$ Hitchin system on a genus 2 curve. The technique developed in \cite{Gaw} certainly enables one to give a theta-function formula for Hitchin trajectories. By means of a different Lax representation, it has been done by Krichever for $GL_n$ and $SL_n$ Hitchin systems of arbitrary rank and genus, as it is described in \refSS{ISM}. As for the action-angle variables, they  are found for all simple classical structure groups and hyperelliptic curves of arbitrary genus, as it is described in \refSS{SoV}.

%%%%%%%%%%%%%%%%%%%%%%%%%%%%%%%%%%%%%%%%%%%%%%%%%%%%%%%%%%%%%%%%%%%%%%
\subsection{Inverse spectral method}\label{SS:ISM}
\subsubsection{Lax operator}\label{SS:Lax}
In the case $\g=\gl(n)$ we define the Hitchin system by its Lax matrix $L$ which is a $n\times n$ matrix-valued meromorphic function on $\Sigma$, depending on parameters called times. It is assumed that $L$ has a stationary (independent of times) pole divisor $D_{st}=\sum_{P\in\Sigma}m_PP$  ($m_P\ge 0$; $m_P=0$ except for a finite subset of $\Sigma$), and $ng$ simple poles $\{\ga_s|s=1,\ldots,ng\}$ depending on times. Below, $D_{st}$ is assumed to be a divisor of a holomorphic differential $\xi$ on $\Sigma$ (this property distinguishes Hitchin systems in the class of all integrable systems with the spectral parameter on a Riemann surface). In a neighborhood of a point $\ga_s$, for every $s$, $L$ is assumed to have an expansion of the form
\begin{equation}\label{E:Lax}
   L(z)=\frac{\b_s\a_s^T}{z-z_s}+L_{0,s}+O(z-z_s)
\end{equation}
where $\a_s,\b_s$ are $n\times 1$ matrices, $\a_s^T\a_s=0$, $\b_s^T\a_s=0$, $\exists \varkappa_s\in\C : L_{0,s}\a_s=\varkappa_s\a_s$, the $^T$ denotes transposition, $z$ is a local coordinate in the neighborhood of $\ga_s$, $z_s=z(\ga_s)$. The elements $\{\ga_s|s=1,\ldots,ng\}$, $\{\a_s|s=1,\ldots,ng\}$ are referred to as Tyurin parameters. They come from the parametrization of semi-stable holomorphic vector bundles due to A.N.Tyurin \cite{Tyur65,Tyur66}, in the following way.  By Riemann-Roch, a semi-stable holomorphic rank $n$ degree $ng$ vector bundle has an $n$-dimensional space of holomorphic sections. Choose a base in this space. The corresponding sections generically have $ng$ points of linear dependence (these are $\{\ga_s|s=1,\ldots,ng\}$), every one being assigned with a set of coefficients of linear dependence which we organise in the matrices $\a_s$. Let $\Phi$ be a Higgs field. Then $\Phi/\xi$ is an endomorphism of the space of meromorphic sections of the bundle. We claim that $L$ is nothing but the matrix of the endomorphism $\Phi/\xi$ in the above base \cite{Kr_Lax}.
%%%%%%%%%%%%%%%%%%%%%%%%%%%%%%%%%%%%%%%%%%%%%%%%%%%%%%%%%%%%%%%
\subsubsection{Hierarchy of Lax equations}\label{SS:hierarchy}
It is our next step to give the Hitchin flows by Lax equations of the form $\dot L=[M,L]$. Here, $M$ is a meromorphic matrix-valued function on $\Sigma$  assumed to have poles at the same points as $L$, expansions of the form $M(z)=\frac{\mu_s\a_s^T}{z-z_s}+M_{0,s}+O(z-z_s)$ at the points $\ga_s$, with no additional relations between $\a_s$, $\mu_s$, and $M_{0,s}$.
Plugging the expansions for $L$ and $M$ to the Lax equation, we obtain the following important evolution equations of Tyurin parameters:
\begin{equation}\label{E:Tyur_move}
  \dot z_s=-\mu_s^T\a_s,\quad \dot\a_s=-M_{0,s}\a_s+\kappa_s\a_s
\end{equation}
where $\kappa_s\in\C$ (observe that signs on the right hand side of the last relation are invariant with respect to the simultaneous change of order of $\a_s$ and $\b_s$ in \refE{Lax}, and of $L$ and $M$ in the Lax equation).

According to  \cite{Kr_Lax}, every $M$-operator is assigned to a certain time, and defines the correspondent flow. The times are enumerated by the triples  $a=(P,m,k)$, $P\in {\rm supp}(D_{st})$, $m>-m_P$, $k\in\Z_+$. The following theorem states the existence and commutativity of those flows.
\begin{theorem}[\cite{Kr_Lax}]\label{T:hierarchi}
  $1^\circ$. For every triple $a=(P,m,k)$ of the above form, there exists the unique $M$-operator $M_a$ such that it has the only pole $P$ in ${\rm supp}(D_{st})$, $M_a(w)=w^{-m}L^k(w)+O(1)$ ($w$ being a local coordinate in a neighborhood of $P$), and $M_a(P_0)=0$ ($P_0$ being some point in $\Sigma$, the same for all $a$).

 $2^\circ$. The equations $\partial_aL=[M_a,L]$ give a family (called hierarchy) of commuting flows.
\end{theorem}
There is developed in \cite{Kr_Lax} a symplectic theory of the systems in question, which implies in particular that the flows in the \refT{hierarchi} exactly correspond to the Hitchin Hamiltonians.
%%%%%%%%%%%%%%%%%%%%%%%%%%%%%%%%%%%%%%%%%%%%%%%%%%%%%%%%%%%%%%%%%%%%%%
\subsubsection{Spectral transform}\label{SS:Spec_tr}
Consider the eigenvalue problem
\begin{equation}\label{E:spec_L}
   (L(q,t)-\l)\psi(Q,t)=0,\ q\in\Sigma ,\ Q=(q,\l).
\end{equation}
Define the spectral curve of $L$ by the equation
\begin{equation}\label{E:spec_curve}
   \det(L(q,t)-\l)=0,\ q\in\Sigma,
\end{equation}
and denote it by $\widehat{\Sigma}$. Let $\widehat{g}$ be its genus, $\pi : \widehat{\Sigma}\to\Sigma$ be the natural projection.
Normalize $\psi$ by the condition
\begin{equation}\label{E:norm}
    \sum_{i=1}^{ n} \psi_i=1.
\end{equation}
Then $\psi(Q)$ becomes a meromorphic function on $\widehat\Sigma$. Its pole divisor plays a special role and is referred to as a dynamical divisor of the problem. It depends on time. We denote it by $D(t)$.

%%%%%%%%%%%%%%%%%% Igor
It follows from the equation $\partial_aL = [M_a,L]$ that, if $\psi$ is an eigenvector of $L$, then $(\partial_a -M_a)\psi$ is also
an eigenvector. Therefore,
\[
  (\partial_a-M_a)\psi(Q,t) = f_a(Q,t)\psi(Q,t)
\]
where $f_a(Q,t)$ is a scalar meromorphic function on $\Sigma$.
The vector-function
\begin{equation}\label{E:gauge_t}
 \widehat{\psi}(Q,t) = \phi(Q,t)\psi(Q,t);\quad \phi(Q,t)=\exp \left(
 -\int_{0}^{t_a}f_a(Q,\tau)d\tau
 \right)
\end{equation}
satisfies the equations
\[
 L(q,t)\widehat{\psi}(Q,t) = \l\widehat{\psi}(Q,t);\quad (\partial_a-M_a)(q, t))\widehat{\psi}(Q,t) = 0.
\]
It turns out that the pole divisor $D(t)$ under the gauge transform \refE{gauge_t} gets transformed to a time-independent divisor $D=D(0)$
of poles of $ \widehat{\psi}$ (see the explanation in the end of \refSS{ISM_obstr}). All the time dependence of  $\widehat\psi$ is
encoded in the form of its essential singularities, which it acquires at the constant poles of $f_a$.
%%%%%%%%%%%%%%%%%%% Igor
\begin{theorem}[\cite{Kr_Lax}]\label{T:Baker-Akh}
The $\widehat{\psi}$ possesses the following properties:

$1^\circ$. it is meromorphic except at the preimage of the divisor $(\xi)$. The degree of its divisor of poles is equal to $\widehat{g}+n-1$.

$2^\circ$. Let $P\in{\rm supp}(D_{st})$. Then $\widehat{\psi}$ has the following expansion in the neighborhood of $P^l$:
\begin{equation}\label{E:exp_p}
   \widehat{\psi}_l(w,t)=\xi_l(w,t)\exp\left(\sum_{n=1}^\infty\sum_{m\ge -m_P} t_{P,m,n}w^{-m}\l_l(w)^n\right)
\end{equation}
where $w$ is a local parameter on $\widehat{\Sigma}$ in the neighborhood of $P^l$ (generically, it is independent of $l$), $\xi_l(w,t)$ is a (vector-valued) Taylor series in $w$, $\l_l(w)$ is the characteristic root of $L(P(w),t)$ corresponding to the $l$th covering sheet.

$3^\circ$.
In a neighborhood of $P_0^l$ the $\psi$ has an expansion of the form
\begin{equation}\label{E:exp_p0}
   \widehat{\psi}_{0,l}(w,t)=\xi_{0,l}(w,t)\exp\left(\sum_{m=1}^\infty t_{l,m}w^{-m}\right)
\end{equation}
where $w$ is a local parameter in the neighborhood of $P_0^l$, $\xi_{0,l}(w,t)$ is a (vector-valued) Tailor series in $w$ such that $\xi_{0,l}^i\big|_{w=0}=\d^i_l$.
\end{theorem}
The function $\widehat{\psi}$ possessing the above properties is called a Baker--Akhieser function associated with the divisor $D$.

%%%%%%%%%%%%%%%%%%%%%%%%%%%%%%%%%%%%%%%%%%%%%%%%%%%%%%%%%%%%%%%%%%%%%%
\subsubsection{$\theta$-function solution by inverse spectral method}
In the remainder of the section we follow the standard prescriptions of the algebraic--geometric version of the inverse spectral method, due to I.Krichever \cite{KP2,KP3}, specified for Hitchin systems in \cite{Kr_Lax}.
For technical reasons, we enumerate the sheets of the covering $\pi$ in an arbitrary way; what follows does not depend on this enumeration. For a function $\widehat{\psi}$ on $\widehat{\Sigma}$ we denote its restriction on the $l$th sheet by $\widehat{\psi}_l$. Similarly, for a $P\in\Sigma$ let $P^l$ stay for its preimage on the $l$th sheet.
\begin{theorem}[\cite{Kr_Lax}]\label{T:LM_restore}
Given a Baker--Akhieser function $\widehat{\psi}$ satisfying the conditions of \refT{Baker-Akh}, for every $a$ there exist the unique matrices $L$, $M_a$ such that
\begin{equation}\label{E:eigen_eq}
    L\widehat{\psi}=\l\widehat{\psi},\quad (\partial_a-M_a)\widehat{\psi}=0
\end{equation}
where $\partial_a=\partial/\partial t_a$, $L$ satisfies the definition of the Lax  matrix, $M_a$ is an $M$-operator, and the equation $\partial_aL=[L,M_a]$ is fulfilled.
\end{theorem}
\refT{LM_restore} enables us to explicitly express $L$, $M_a$ in terms of  $\widehat\psi$. Take an open set $U\subset\Sigma$ outside the branch divisor. Then $\pi^{-1}U=\bigcup\limits_{l=1}^n U_l$ where $U_l$ belongs to the $l$th sheet of the covering, and $z$ can serve as a local parameter in $U_l$  for all $l=1,\ldots,n$. Let  $\widehat\psi_l=\widehat\psi\big|_{U_l}$, $\Psi$ be the matrix formed by the vectors $\widehat\psi_l$ ($l=1,\ldots,n$) as columns. The $\Psi$ depends on the order of sheets of the covering, however the final result will not depend on this ambiguity. The first of the relations \refE{eigen_eq} can be written down as $L\Psi=\Psi\Lambda$ where $\Lambda={\rm diag}(\l_1,\ldots,\l_n)$ is the spectrum of $L$. Similarly,
$(\partial_{P,m,n}-M_{P,m,n})\Psi=0$. We finally have
\begin{equation}\label{E:LM_final}
  L=\Psi\Lambda\Psi^{-1},\quad M_{P,m,n}=-\partial_{P,m,n}\Psi\cdot\Psi^{-1}.
\end{equation}
It is already conventional that the Baker--Akhieser functions can be expressed via $\theta$-functions \cite{KP2,KP3,Kr_Seraya_Kn}. This way, we can explicitly express $L$ and $M$ in terms of the spectral curve and the divisor $D$ on it, via $\theta$-functions.

%%%%%%%%%%%%%%%%%%%%%%%%%%%%
\subsection{Separation of variables for Hitchin systems with classical structure groups}\label{SS:SoV}
%%%%%%%%%%%%%%%%%%%%%%%%%%%%%%%%%%%%%%%%%%%%%
\subsubsection{Spectral curve}
To define a spectral curve, we first choose and fix a holomorphic 1-form $\xi$ on $\Sigma$. Given a Higgs field $\Phi$, we define its spectral curve by means of the equation
\[
   \det\left(\l-\frac{\Phi}{\xi}\right)= \l^d+\sum_{j=1}^n \l^{d-d_j} r_j =0
\]
where $d$ is the rank of the bundle, $n$ is the rank of the Lie algebra $\g$. Let $\rho_j$, $j=1,\ldots, n$ denote the basis invariant polynomials of the Lie algebra $\g$, $\d_j=\deg \rho_j$. For $\g=\gl(n)$ or $\g$ a classical simple Lie algebra, we have $r_j=\rho_j(\Phi/\xi)$, $d_j=\d_j$ except for $\g={\mathfrak so}(2n)$, $j=n$ when  $r_{2n}=\rho_n(\Phi/\xi)^2$, $d_n=2\d_n$. Indeed, in that case $r_n=\det(\Phi/\xi)$, and $\rho_n(\cdot)={\rm Pf}(\cdot)$.

For any divisor $D$ let ${\mathcal O}(D)$ stay for the space of meromorphic functions $f$ on the spectral curve, such that $(f)+D\ge 0$. Below, fix $D$ to be divisor of the differential $\xi$, i.e. $D= (\xi)$.
\begin{proposition}\label{P:invariants}
$\rho_j\in {\mathcal O}(\d_jD)$ ($j=1,\ldots,n$).
\end{proposition}
Introduce also the following short notation for the characteristic polynomial of $\Phi/\xi$:
\begin{equation}\label{E:short_spec}
  R(\l,P)=\l^d+\sum_{j=1}^n \l^{d-d_j} r_j(P),\quad P\in\Sigma .
\end{equation}
\begin{proposition}\label{P:corr_Phi}
$[\Phi ]\to R$ is a one-to one correspondence between classes of gauge equivalence of Higgs fields and their characteristic polynomials.
\end{proposition}
\begin{proof}
Obviously, the correspondence is a continuous injection. We will prove that the dimension of its image is equal to the dimension of the class of characteristic polynomials. Indeed, in the case $\g$ is simple, $\dim\{ [\Phi ]\}=\dim\g\cdot(g-1)$ \cite{Hitchin}. On the other hand side, $\dim\{ R\}=\sum_{j=1}^n\dim {\mathcal O}(\d_jD)$. By Riemann--Roch, $\dim {\mathcal O}(\d_jD)=(2\d_j-1)(g-1)$, and by the wellknown Kostant identity $\sum_{j=1}^n(2\d_j-1)=\dim\g$. For $\g=\gl(n)$, similarly, $\sum_{j=1}^n\dim{\mathcal O}(\d_jD)=n^2(g-1)+1=\dim \{ [\Phi] \}$. Hence our claim is true.
\end{proof}

Given a certain class of base curves,  propositions \ref{P:invariants} and \ref{P:corr_Phi} provide a description of the family of spectral curves corresponding to all possible Higgs fields $\Phi$. In \cite{Sh_FAN_2019, BorSh} it has been done for hyperelliptic base curves, as follows.

Let $\Sigma : y^2=P_{2g+1}(x)$, \ $\displaystyle \xi=\frac{dx}{y}$ \ \ {\Large (}then $D=2(g-1)\cdot\infty${\Large )}.
\begin{proposition}[\cite{Sh_FAN_2019}]
A base in ${\mathcal O}(\d_jD)$ is formed by the functions $1,x,\ldots,x^{\d_j(g-1)}$, and $y,yx,\ldots,yx^{(\d_j-1)(g-1)-2}$.
\end{proposition}
\begin{proposition}[\cite{Sh_FAN_2019}]
(The affine part of) the spectral curve for a Hitchin system of  types $A_n$, $B_n$, $C_n$ is a full intersection of the two surfaces in $\C^3$: $y^2=P_{2g+1}(x)$ and $R(\l,x,y,H)=0$ where
\begin{equation}\label{E:sp_curve_eq}
   R=\l^d+\sum_{j=1}^n \left(\sum_{k=0}^{\d_j(g-1)}H_{jk}^{(0)}x^k+ \sum_{s=0}^{(\d_j-1)(g-1)-2}H_{js}^{(1)}x^sy\right)\l^{d-d_j},
\end{equation}
$H_{jk}^{(0)}$, $H_{js}^{(1)}\in\C$. For $\g=\so(2n)$ the bracket in the last summand ($j=n$) must be squared.
\end{proposition}

%%%%%%%%%%%%%%%%%%%%%%%%%%%%%%%%%%%%%%%
\subsubsection{Hamiltonians}
Here we define the Hamiltonians of a Hitchin system.

The above coefficients $H_{jk}^{(0)}$, $H_{js}^{(1)}$ are obviously the functions of $\Phi$. These functions are referred to as Hamiltonians of the hyperelliptic Hitchin system.

In case of an arbitrary base curve pick up a base in ${\mathcal O}(\d_jD)$ and expand $\rho_j(\Phi/\xi)$ over this base. The coefficients of the expansion, as functions of $\Phi$, we call Hitchin Hamiltonians.

Original procedure due to N.Hitchin prescribes to pick up a base in $H^0(\Sigma,{\mathcal K}^{\d_j})$ rather than in ${\mathcal O}(\d_jD)$. However, these two procedures are equivalent because $H^0(\Sigma,{\mathcal K}^{\d_j})\simeq {\mathcal O}(\d_jD)\xi^{\otimes \d_j}$, $\forall j$.
%%%%%%%%%%%%%%%%%%%%%%%%%%%%%%%%%%%%%%%%%%%%%%
\subsubsection{Separating variables}
Every curve can be given by a set of points it is coming through. We will show that these points serve as separating coordinates of Hitchin systems. The idea goes back to \cite{GNR,BT2}. However it is realized in full only in \cite{Aint_sys,Sh_FAN_2019} by focusing on a certain class of base curves.

Let $h=\dim\{ [\Phi] \}$; $h$ is the same as the number of degrees of freedom of the Hitchin system. As well, the number of coefficients in the equation of the spectral curve is equal to $h$, as it is proven in the course of the proof of \refP{corr_Phi}. Hence, it is sufficient to give $h$ points of the spectral curve to express the coefficients via them. Denote these points by $\ga_i=(x_i, y_i, \l_i) ~(i = 1, \ldots, h)$, where $y_i^2 = P_{2g+1}(x_i)$,  $\forall i$. Let $H$ be the full set of coefficients in \refE{sp_curve_eq}. Then $H$ can be expressed via $(x_i, y_i, \l_i)$, $(i = 1, \ldots, h)$ from the equations
\begin{equation}\label{E:sep}
  R(x_i, y_i, \l_i,H)=0,\quad i=1,\ldots,h
\end{equation}
where every equation contains only one triple $(x_i, y_i, \l_i)$. A system of equations satisfying the last requirement is referred to as a system of separating equations, and $(x_i, y_i, \l_i)$, $(i = 1, \ldots, h)$ as  separating variables.

We stress that for $\g=\gl(n), \sln(n), \spn(2n), \so(2n+1)$ the system of equations \refE{sep} is linear in $H$ (compare with \refE{sp_curve_eq}), while in the case $\g=\so(2n)$ it is quadratic. Nevertheless, for $n=2$, it can be effectively resolved in the last case (P.Borisova, \cite{BorSh}).

%%%%%%%%%%%%%%%%%%%%%%%%%%%%%%%%%%%
\subsubsection{Symplectic form and Poisson structure}\label{SS:Sym_Po}
%Derive the symplectic form from the differential $\l dz$.
Here we define the symplectic form of a Hitchin system, and express it via separating coordinates.

In \cite{GNR} the action--angle variables $(I,\a)$ are defined by
\[
   \l dz=\sum_{a=1}^{h} I^a\Omega_a,\quad
   \a_a=\sum_{i=1}^h \int^{\ga_i}\Omega_a,
\]
where $z=\int^{\ga}\widehat{\xi}$ is a (quasiglobal) coordinate on $\widehat{\Sigma}$ ($\widehat{\xi}$ is the pull-back of $\xi$ to $\widehat{\Sigma}$), $\Omega_a$ are normalized Abel/Prym differentials,
$h=\dim\g(g-1)$ for $\g=SO(n)$, $Sp(2n)$, and  $h=n^2(g-1)+1$ for $\g=\gl(n)$.

For the symplectic form it follows (by definition and variation of the above integrals in $\ga_i$)
 \[
   \w=\sum_{a=1}^h \d I^a\wedge\d\a_a = \sum_{a=1}^h \d I^a\wedge \sum_{i=1}^h\Omega_a(\ga_i)
 \]
 ($\d$ means variation in $\{\ga_i \}$).
Also,
\[
 \l dz=\sum_{a=1}^{h} I^a\Omega_a\  \Longrightarrow\ \l_i\d z_i=\sum_{a=1}^{h} I^a\Omega_a(\ga_i)\ \Longrightarrow\ \d\l_i\wedge\d z_i = \sum_{a=1}^{h} \d I^a\wedge \Omega_a(\ga_i),\ \forall i
\]
where $z_i=\int^{\ga_i}\widehat{\xi}$.
For hyperelliptic curves $dz=\widehat{\xi}$, $\xi=\frac{dx}{y}$, and it follows that
\begin{equation}\label{E:Sympl2}
 \w= \sum_{i=1}^h \d\l_i\wedge \d z_i= \sum_{i=1}^h \d\l_i\wedge  \frac{\d x_i}{y_i}.
\end{equation}
Hence the Poisson structure is given by setting
\begin{equation}\label{E:Poisson}
       \{ \l_i,x_j\}=\d_{ij} y_i,
\end{equation}
and all other brackets vanish.

%%%%%%%%%%%%%%%%%%%%%%%%%%%%%%%%%%%%%%%%%%%%%%%%%%%%%%%%%%%%%%%%%%%%%
\subsubsection{Integrability}
The Hamiltonians defined by the system of equations \refE{sep}, and the Poisson bracket   \refE{Poisson} give a Hitchin system in full in  terms of the set of points $ \{ (x_i, y_i, \l_i) \ |i = 1, \ldots, h \}$.
It is our next goal to prove the integrability of Hitchin systems in these terms.

More generally, consider a system of equations
\begin{equation}\label{E:non_sys}
F_i(H_1,\ldots,H_h,\l_i,x_i)=0,\quad i=1,\ldots,h
\end{equation}
where $F_i$ are smooth functions, $H_j$,  $i,j=1,\ldots,h$ are defined by the relations \refE{non_sys} as functions of $(\l_1,\ldots,\l_h,x_1,\ldots,x_h)$. It is important that for any $i=1,\ldots,h$ the function $F_i$ explicitly depends on only one pair of variables $\l_i,x_i$ (however it may depend on the remainder of variables via $H_1,\ldots,H_h$).

Consider a Poisson bracket on $\C^{2n}$ of the form
\[
\{ f,g\}=\sum\limits_{j=1}^np_j\left(\frac{\partial f}{\partial \l_j}\frac{\partial g}{\partial x_j}-\frac{\partial g}{\partial \l_j}\frac{\partial f}{\partial x_j}\right)
\]
where $p_j=p_j(\l_j,x_j)$ are smooth functions in only one pair of variables, and  $p_k=0$ if $\rank\frac{\partial(F_1,\ldots,F_{k-1},F_{k+1}\ldots,F_h)} {\partial(H_1,\ldots,H_{k-1},H_k,H_{k+1}\ldots,H_h)}<n-1$.
\begin{proposition}[\cite{Aint_sys}]\label{P:comm}
$H_1,\ldots,H_h$ commute with respect to any Poisson bracket of the above form.
\end{proposition}
Integrability of a Hitchin system is directly implied by \refP{comm} because the Hitchin Hamiltonians and the Poisson structure \refE{Poisson} are precisely of the required form (with $F_j=R$, and $p_j=y_j$, $\forall j$).
%%%%%%%%%%%%%%%%%%%%%%%%%%%%%%%%%%%%%%%%%%%%%%%%%%%%%%%%%%%%%%%%%%%%%
\subsubsection{Method of generating functions and Darboux coordinates}
Our next goal is to find the coordinates $\phi_j$ ($j=1,\ldots,h$) conjugate to Hamiltonians. These are coordinates in a generic fiber of the Hitchin fibration. They can be found out in a standard way by means of the \emph{generating functions} technique \cite{Arnold}. In frame of separation of variables a generating function is a function of the form $S=\sum\limits_{i=1}^h S_i$ where $S_i=S_i(z_i,H)$, $z_i=\int\limits^{\ga_i}\widehat{\xi}$, as in \refSS{Sym_Po}. Observe that $(\l_i,z_i)$ are canonical coordinates \refE{Sympl2}.

The $S_i$ is defined by the relation $R(x_i,y_i,s_i,H)=0$ where
$
   s_i(z_i,H)=\frac{\partial S_i}{\partial z_i}
$, then
$
   \frac{\partial R}{\partial\l} \frac{\partial s_i}{\partial H_j}+\frac{\partial R}{\partial H_j}=0,
$
which implies
\[
          \frac{ \partial s_i}{\partial H_j}=-\frac{{\partial R}/\partial H_j}{{\partial R}/\partial\l }.
\]
By the above notation $S_i=\int\limits^{\ga_i}s_i(z_,H)dz$. We can write it in the form $S_i=\int\limits^{\ga_i}s_i\frac{dx}{y}$, thus by $\frac{dx}{y}$ we actually mean here the pull-back of this differential to $\widehat{\Sigma}$, i.e. the $\widehat{\xi}$.

By a general statement of the generating function technique, the symplectically dual coordinates to $H_j$ are as follows:
\[ \phi_j=\frac{\partial S}{\partial H_j},
\]
i.e.  $\{ H_j,\phi_j\ |\ j=1,\ldots,h\}$ are Darboux coordinates for the Hitchin system. In the frame of separation of variables $\frac{\partial S}{\partial H_j}=\sum_i \frac{\partial S_i}{\partial H_j}=\sum_i \int\limits^{\ga_i}\frac{\partial s_i}{\partial H_j}\frac{dx}{y},
$
and finally
\begin{equation}\label{E:odinar}
  \phi_j=-\sum_{i=1}^{h} \int\limits^{\ga_i}\frac{{\partial R}/\partial H_j}{{\partial R}/\partial\l }\frac{dx}{y}.
\end{equation}
Some details of the calculation go back to \cite{Hur}, for the remainder of details we refer to \cite{Sh_FAN_2019}; see also \cite[eq. (4.61)]{Kr_Lax}.

In the case the equation of a spectral curve is linear in $H$-coordinates, the Darboux property follows also from the results of \cite{BT2,DT}. For the root system $D_n$ these results do not work.
\begin{proposition}[\cite{BorSh}]\label{P:h_diff1}
In the case of Hitchin systems with the structure group $GL(n)$ the differentials $\frac{{\partial R}/\partial H_j}{{\partial R}/\partial\l }\frac{dx}{y}$ ($j=1,\ldots,h$) form a base of holomorphic differentials on the spectral curve. For the systems with the structure groups $SL(2)$, $SO(2n+1)$, $Sp(2n)$ they form a base of holomorphic Prym differentials on the spectral curve with respect to the involution $\l\to -\l$, while for the systems with the structure group $SO(2n)$ they form a base of holomorphic differentials on the normalization of the spectral curve. In the case of structure group $SL(n)$ the base of holomorhic differentials on the spectral curve consists of the differentials $\frac{{\partial R}/\partial H_j}{{\partial R}/\partial\l }\frac{dx}{y}$, and of $g$ basis holomorphic differentials pulled back from the base curve.
\end{proposition}

%%%%%%%%%%%%%%%%%%%%%%%%%%%%%%%%%%%%%%%%%%%%%%%%%%%%%%%%%%%%%%%%%%%%%
\subsubsection{$\theta$-function solution}\label{SS:theta_sol}
In the coordinates $(H,\phi)$, the Hamiltonian equations read as $\dot{H}=0$, $\dot\phi=H$, hence the trajectories of the system are lines, i.e. straight windings of fibers of the Hitchin map. We address here the problem of finding out the trajectories in the coordinates $\{ (x_i,y_i,\l_i)\}$ representing $\rm{Sym}^h\Sigma$. In the case when $\g=\gl(n)$, the transform $\ga \to \phi$ given by \refE{odinar} coincides with the Abel transform up to a constant linear transformation coming from the passing from the integrands in \refE{odinar} to normalized holomorphic differentials (observe that $h=n^2(g-1)+1=genus (\widehat{\Sigma})$ in that case). This difference does not matter in this section. We will ignore it and regard to \refE{odinar} as to Abel transform. Then it is a natural idea to transfer the above straight windings to $\rm{Sym}^h\Sigma$ by means of the Jacobi inverse map. Looking ahead (\refSS{SoV_obstr}), we notice that, for other classical Lie algebras, there is a certain obstruction for such approach.

The setting of the Jacobi inversion problem is as follows: given a point $\phi$ of the Jacobian of a genus $h$ curve, to find out the points $\ga_1,\ldots,\ga_h$ of the curve such that $\mathcal{A}(\ga_1+\ldots+\ga_h)=\phi$, where $\mathcal{A}$ is the Abel transform. We refer to $\mathcal{A}^{-1}$ as to Jacobi inverse map. In principle, solution of the problem is given by the Riemann theorem: $\ga_1,\ldots,\ga_h$ are zeroes of the function $\theta(\mathcal{A}(\ga)-\phi-K)$ where $K$ is the vector of Riemann constants, provided the last function does not vanish identically.

In \cite{Dubr_theta} B.A.Dubrovin set the problem of finding out the points $\ga_1,\ldots,\ga_h$ analytically. For this purpose, he proposed, first, to take a meromorphic function $f$ on the curve. Then $\s_f=\sum_{j=1}^{h}f(\ga_j)$ is a symmetric function on the curve, hence a meromorphic function of $\phi$, that is an Abelian function. As such, it can be expressed via the Riemann theta-function. Assume, the curve is represented as a covering of $\C P^1$. Then for the functions $f_k(P)=z^k$ ($z=\pi(P)$, $\pi$ is the covering map, $k=1,\ldots,h$) the desired expressions can be given explicitly. Observe that $\s_{f_k}$ (or just $\s_k$, for short) is a Newton polynomial in $z_1,\ldots, z_h$, hence the values of $\s_1,\ldots, \s_h$ determine the set $z_1,\ldots, z_h$ up to a permutation.

Here, we realize this program for the spectral curve, and the covering map $\pi : P\to x$, $P=(x,y,\l)$. As a result, we will obtain the set $x_1,\ldots, x_h$ up to a permutation, which gives a point of $\rm{Sim}^h\widehat{\Sigma}$ locally, i.e. we will locally obtain the trajectories.

We start from the relation due to Dubrovin \cite{Dubr_theta}
\begin{equation}\label{E:Dubr4}
  \s_k(\phi)=const-\res_{x=\infty} x^kd\ln\theta(\mathcal{A}(P)-\phi-K).
\end{equation}
Since $(d\mathcal{A})_i=\w_i$ we have
\[
  d\ln\theta(\mathcal{A}(P)-\phi-K)=\sum_{i=1}^h (\partial_i\ln\theta(\mathcal{A}(P)-\phi-K))\w_i
\]
where $\w_i=\frac{{\partial R}/\partial H_i}{{\partial R}/\partial\l }\frac{dx}{y}$, $\partial_i$ means a derivative in the $i$th argument. Choosing  $x=\infty$ as a base point of the Abel transform, we can regard to $\mathcal{A}(P)$ as to a small quantity, and expand $(\ln\theta(\mathcal{A}(P)-\phi-K))_i$ to a Tailor series. Then we only have to select the terms of order $z^{2k-1}$ in thus obtained expansion, where $z$ is a local parameter in the neighborhood of $x=\infty$.
Indeed, multiplied by $x^k=z^{-2k}$ it will give the required residue in \refE{Dubr4}. Here is the final result:
\[
  \s_k(\phi)=const- \sum_{i=1}^h\sum_{|j|\le 2k-1} \varkappa_i^jD^j\partial_i\ln\theta(-\phi-K)
\]
where $j=(j_1,\ldots,j_h)$, $|j|=j_1+\ldots+j_h$,
\[
  D^j=\frac{1}{j_1!\ldots j_h!}\frac{\partial^{|j|}}{\partial\phi_1^{j_1}\ldots\partial\phi_h^{j_h}},
  \quad \varkappa_i^j=\sum_{l_i+\sum_{s=1}^h\sum_{p=1}^{j_s} l_{sp}=2k-1} \varphi_i^{(l_i)}\prod_{s=1}^h\prod_{p=1}^{j_s}\frac{\varphi_s^{(l_{sp}-1)}}{l_{sp}},
\]
$l_s$ and $\varphi_s^{(l_s)}$ are defined from the relation $\A_s(P)=\sum_{l_s\ge 1}\frac{\varphi_s^{(l_s)}}{l_s}z^{l_s}$ ($P=P(z)$).

By plugging $\phi=Ht+\phi_0$ we obtain a full set of symmetric functions of $x_1\ldots,x_h$ along trajectories.
%%%%%%%%%%%%%%%%%%%%%%%%%%%%%%%%%%%%%%%%%%%%%%%%%%%%
\subsection{Obstructions in the case of simple groups}\label{SS:obstr}

\subsubsection{Obstruction for the inverse spectral method}\label{SS:ISM_obstr}
The analog of \refT{hierarchi} in the case of classical simple Lie algebra $\g$, except for $\g=\spn(2n)$, has been proven in \cite{Sh_Hierarch} (see also \cite{Sh_DGr}). As soon as the existence of the commuting hierarchy is established, the problem arises of resolving that hierarchy by means of the inverse spectral method, similar to \refSS{ISM}. However, there is an obstruction on this way, which we will demonstrate in the case when $\g$ is the Lie algebra of an orthogonal group. In this case,
\[
  L_s(z)=\frac{\b_s\a_s^T-\a_s\b_s^T}{z-z_s}+L_{0s}+\ldots .
\]
To obtain the (global) function $\psi$ on $\widehat\Sigma$ whose values are eigenvectors of $L$ at every point, we first seek for the eigenvectors either pointwise, or locally, and then apply a normalization (for instance, \refE{norm}). By \cite[Lemma 5.12]{Sh_DGr}, in a neighborhood of $\ga_s$, there is an eigen-vector-valued function $\psi_s$ of the form $\psi_s(z,\l)=\frac{a_{s,\l}}{z-z_{s,\l}}+O(1)$ where $a_{s,\l}=\nu_{s,\l}\a_s$, $\nu_{s,\l}\in\C$. Away from the branch points, $z-z_s$ can serve as a local parameter on all sheets of the covering over the neighborhood of $\ga_s$, hence we can suppress $\l$ in the notation of the Laurent expansions in $z-z_{s,\l}$. For brevity, we also suppress $\l$ in the notation for $\psi(z,\l)$ and $a_{s,\l}$. With such notations, the above expression reads as $\psi_s(z)=\frac{a_s}{z-z_s}+O(1)$. A further normalization may change the order of the function at $\ga_s$, but the coefficient $a_s$ of the leading term is preserved up to scaling. Observe, after a normalization $\psi$ will not need also the index $s$,  so that in a neighborhood of $\ga_s$ we obtain $\psi=a_s(z-z_s)^m+O((z-z_s)^{m+1})$, for some $m\in\Z$. It is an important property proven in \cite{Kr_Lax} (see \cite{Sh_DGr} for a proof in the simple case) that an eigenvalue of $L$ is holomorphic at the points $\ga_s$, as a function of $z$. For this reason, the equation $L\psi=\l\psi$ implies $(\b_s\a_s^T-\a_s\b_s^T)a_s=0$ (otherwise the orders of the left hand side and the right hand side of the equation are different), hence $\b_s(\a_s^Ta_s)-\a_s(\b_s^Ta_s)=0$. In a generic position $\a_s$ and $\b_s$ are linear independent. Hence $\a_s^Ta_s=\b_s^Ta_s=0$. Next, write down the identity $(\partial_a-M_a)\psi=f_a\psi$ in the form
\begin{equation}\label{E:Mafa}
  \left(\partial_a -\frac{\mu_s\a_s^T-\a_s\mu_s^T}{z-z_s}+M_{0s}+\ldots          \right)(a_s(z-z_s)^m+\ldots)=f_a(a_s(z-z_s)^m+\ldots)
\end{equation}
where dots denote the higher order terms. On the left hand side, the leading term is equal to $a_s\partial_a(z-z_s)^m+\frac{\a_s\mu_s^T}{z-z_s}(z-z_s)^m=(a_sm(-\partial_az_s)+ \a_s\mu_s^T a_s)(z-z_s)^{m-1}$. By \refE{Tyur_move}, $-\partial_az_s=\mu_s^T\a_s$. Since $a_s=\nu_s\a_s$, we have $a_s\mu_s^T\a_s=\a_s\mu_s^T a_s$, hence the leading term is equal to $-(m+1)(\partial_az_s) a_s(z-z_s)^{m-1}$. Then \refE{Mafa} implies $f_a=\frac{f_{-1}}{z-z_s}+O(1)$, and $f_{-1}=-(m+1)(\partial_az_s)$.
Hence
\[
  f_a(z)=\frac{f_{-1}}{z-z_s}+O(1) = (m+1)\frac{-{\partial_a z_s}}{z-z_s}+O(1) =  (m+1) \partial_a\log(z-z_s)+O(1).
\]
We obtain $e^{ -\int \limits_0^{t_a} f_a dt}=(z-z_s)^{-(m+1)}\cdot O(1)$, hence $\widehat{\psi}=e^{ -\int \limits_0^{t_a} f_a dt}\psi=\frac{a_s}{z-z_s}+O(1)$. We conclude that $\widehat{\psi}$ has poles at moving points $\ga_s$, i.e. one of main postulates of the inverse spectral method, namely, immovability of the pole divisor of the Baker--Akhieser function, can't be fulfilled. Observe that this obstruction is absent in the case when $\g=\gl(n)$ considered in \refSS{ISM}, due to the absence of the term $\a_s\mu_s^T$ in the nominator of the analog of \refE{Mafa}. The same calculation as above shows that $e^{ -\int \limits_0^{t_a} f_a dt}=(z-z_s)^{-m}\cdot O(1)$ in that case, hence $\widehat\psi$ is holomorphic at $\ga_s$ in accordance with what was said in \refSS{ISM}.

\subsubsection{Obstruction for getting $\theta$-function solutions by means of separation of variables. Integration in quadratures and numerical integration}\label{SS:SoV_obstr}

Here we present an obstruction for obtaining a theta-function solution of the Hitchin systems of types $B_n$, $C_n$, $D_n$ by technique of \refSS{theta_sol}. The relation \refE{Dubr4} is finally based on the Riemann theorem on zeroes of the Riemann $\theta$-function, in particular, on the fact that the number of zeroes is equal to the genus of the spectral curve. In the case $\g=\gl(n)$ the number $h$ of the points $\ga_1,\ldots,\ga_h$ giving the spectral curve is the same (it also can be characterized as the number of degrees of freedom of the system, the number of independent Hamiltonians, the half dimension of the phase space, and so on). The set of zeroes gives the corresponding point of the phase space in that case. However, in the case of a simple structure group, the number $h$ of the points giving the spectral curve is equal to the dimension of its Prym variety, while the number of zeroes of the Prym $\theta$-function is twice that dimension \cite[Corollary 5.6]{Fay}.

Observe that the above obstruction does not interfere with obtaining a solution in quadratures. Let $\A$ stay for the Abel--Prym map. Its Jacobi matrix is equal to $J=(\frac{\partial\phi_i}{\partial x_k})$. By \refE{odinar}, $\frac{\partial\phi_i}{\partial x_k}=-\frac{{\partial R}/\partial H_i}{{\partial R}/\partial\l }\frac{1}{y}\big|_{x=x_k, y=y_k,\l=\l_k}$. The last expression is explicitly given in the coordinates $(x_k, y_k, \l_k)$, hence we can regard to $J^{-1}$ as a matrix explicitly given in those coordinates. Let $X(t)$ denote a point of the Hitchin trajectory in the same coordinates. Then $X(t)=\A^{-1}(\phi(t))$ where $\phi(t)=\phi_0+Ht$ is the image of that trajectory on the fiber. Then $X'(t)=J^{-1}(\phi(t))H$ which implies
\[
      X(t)=X_0+\left(\int_0^t J^{-1}(\phi(t))dt\right)H.
\]
The last is nothing but a solution to the Hitchin system in quadratures, corresponding to the Hamiltonian $H$, and the initial condition $X_0$. The equation $X'(t)=J^{-1}(\phi(t))H$ admits also a numerical solution. The one with the mesh width $\Delta t$ can be constructed as follows: $X(0)=X_0$, $X(\Delta t)=J^{-1}(X_0)\Delta t$, $X(2\Delta t)=J^{-1}(X(\Delta t))\Delta t$, and so on. Observe that we only use the expression for $J$ in the coordinates $(x,y,\l)$ in the process.

%%%%%%%%%%%%%%%%%%%%%%%%%%%

%%%%%%%%%%%%%%%%%%%%%%%%%%%%%%%%%%%%%%%%%%%%%%%%%%%%
\bibliographystyle{amsalpha}

\end{document}
%%%%%%%%%%%%%%%%%%%%%%%%%%%%%%%%%%%%%%%%%%%%%%%%
%%   THE END
%%%%%%%%%%%%%%%%%%%%%%%%%%%%%%%%%%%%%%%%%%%